\documentclass[journal]{IEEEtran}
\usepackage{amsmath}
\usepackage{amssymb}
\usepackage{amsthm}
\usepackage{color}
\usepackage{xcolor}
\usepackage{caption}
\usepackage{epsfig,latexsym}
\usepackage{float}
\usepackage{fancyhdr}
\usepackage{graphicx}
\usepackage{indentfirst}
\usepackage{mdwmath}
\usepackage{mdwtab}
\usepackage{subfigure}
\usepackage{setspace}
\usepackage{times}
\usepackage{url}
\usepackage{multirow}
\usepackage{algorithm}
\usepackage{algorithmic}
\usepackage[colorlinks=true, allcolors=blue]{hyperref}

\newtheorem{remark}{Remark}
\newtheorem{theorem}{Theorem}

\newtheorem{lemma}{Lemma}

\newtheorem{corollary}{Corollary}

\newtheorem{proposition}{Proposition}

\begin{document}

\title{Performance Analysis of Fluid Antenna System under Spatially-Correlated Rician Fading Channels}
\author{Jiangsheng Huangfu, Zhengyu Song,~\IEEEmembership{Member,~IEEE}, \\
Tianwei Hou,~\IEEEmembership{Member,~IEEE}, Anna Li,~\IEEEmembership{Member,~IEEE}, Yuanwei Liu,~\IEEEmembership{Fellow,~IEEE}, \\
Arumugam Nallanathan,~\IEEEmembership{Fellow,~IEEE}, 
and Kai-Kit Wong,~\IEEEmembership{Fellow,~IEEE} 
\vspace{-8mm}

\thanks{This work was supported in part by the Fundamental Research Funds for the Central Universities under Grant 2024JBMC014 and 2023JBZY012, in part by the National Natural Science Foundation for Young Scientists of China under Grant 62201028, in part by Young Elite Scientists Sponsorship Program by CAST under Grant 2022QNRC001, in part by the Beijing Natural Science Foundation L232041, and in part by the Marie Sk\l{}odowska-Curie Fellowship under Grant 101154499, in part by EPSRC grant numbers to acknowledge are EP/W004100/1, EP/W034786/1, EP/Y037243/1 and EP/W026813/1. ({\em Corresponding author: Tianwei Hou}).}
\thanks{Jiangsheng Huangfu and Zhengyu Song are with the School of Electronic and Information Engineering, Beijing Jiaotong University, Beijing 100044, China (e-mail: 21231312@bjtu.edu.cn; songzy@bjtu.edu.cn).}
\thanks{T. Hou is with the School of Electronic and Information Engineering, Beijing Jiaotong University, Beijing 100044, China, and also with the School of Electronic Engineering and Computer Science, Queen Mary University of London, London E1 4NS, U.K.}
\thanks{Anna Li is with the School of Computing and Communications, Lancaster University, Lancaster LA1 4WA, U.K. (e-mail: a.li16@lancaster.ac.uk).}
\thanks{Yuanwei Liu is with the Department of Electrical and Electronic Engineering, The University of Hong Kong, Hong Kong (e-mail: yuanwei@hku.hk).}
\thanks{A. Nallanathan is with the School of Electronic Engineering and Computer Science, Queen Mary University of London, London, U.K. and also with the Department of Electronic Engineering, Kyung Hee University, Yongin-si, Gyeonggi-do 17104, Korea (e-mail; a.nallanathan@qmul.ac.uk).}
\thanks{Kai-Kit Wong is with the Department of Electronic and Electrical Engineering, University College London, London WC1E 7JE, U.K. and also with the Department of Electronic Engineering, Kyung Hee University, Yongin-si, Gyeonggi-do 17104, Korea (e-mail: kitwong@ieee.org).}
}

\maketitle
\begin{abstract}
Fluid antenna systems (FAS) are among the most promising technologies for the sixth generation (6G) mobile communication networks. Unlike traditional fixed-position multiple-input multiple-output (MIMO) systems, a FAS possesses position reconfigurability to switch on-demand among \(N\) predefined ports over a prescribed space. This paper explores the performance of a single-input single-output (SISO) model with a fixed-position antenna transmitter and a single-antenna FAS receiver, referred to as the Rx-SISO-FAS model, under spatially-correlated Rician fading channels. Our contributions include exact expressions and closed-form bounds for the outage probability of the Rx-SISO-FAS model, as well as exact and closed-form lower bounds for the ergodic rate. Importantly, we also analyze the performance considering both uniform linear array (ULA) and uniform planar array (UPA) configurations for the ports of the FAS. To gain insights, we evaluate the diversity order of the proposed model and our analytical results indicate that with a fixed overall system size, increasing the number of ports, \(N\), significantly decreases the outage performance of FAS under different Rician fading factors. Our numerical results further demonstrate that: $i)$ the Rx-SISO-FAS model can enhance performance under spatially-correlated Rician fading channels over the fixed-position antenna counterpart; $ii)$ the Rician factor negatively impacts performance in the low signal-to-noise ratio (SNR) regime; $iii$) FAS can outperform an $L$ branches maximum ratio combining (MRC) system under Rician fading channels; and $iv)$ when the number of ports is identical, UPA outperforms ULA.
\end{abstract}


\section{Introduction}
After generations of evolution/revolution, wireless communication technologies have continued to reach new heights \cite{5G0}. Among many technological advances, multiple-input multiple-output (MIMO) is arguably the most important of all and its impact has now spanned three generations \cite{MIMO0,MIMO1}. Its ability to increase capacity using the spatial domain beyond time and frequency resources, makes MIMO a timeless technology \cite{MIMO2}. Since the fourth generation (4G), MIMO has been elevated to multiuser MIMO in which a base station (BS) equipped with multiple antennas serves multiple users on the same physical channel to greatly increase system capacity \cite{Wong-mumimo-2002,Wong-mumimo-2003,MU_MIMO2,MU_MIMO1}.

The current fifth generation (5G) has taken another step to have the massive version of MIMO that has the great potential to simplify precoding designs with an excessive number of BS antennas \cite{5G1,5G2}. Contemplating the sixth generation (6G), extra-large MIMO (XL-MIMO) is widely recognized as one of the major enabling technologies going forward \cite{Wang-xlmimo}. There are also promising technologies such as non-orthogonal multiple access (NOMA) \cite{NOMA} and reconfigurable intelligent surfaces (RIS) \cite{RIS} that have generated much interest. Their integration with MIMO has been shown to reduce network interference \cite{MIMO_RIS_NOMA,MIMO_LIS} and enhance communication signals \cite{MIMO_RIS}.

However, over time, the limitations of MIMO have become apparent. Firstly, a large-scale antenna array comes with huge cost of expensive radio frequency (RF) chains and enormous power consumption \cite{MIMO_expensive}. Consequently, hybrid beamforming has been proposed to address the issue of excessive hardware costs \cite{Hybrid_BF}. However, the complexity involved in precoding optimization limits the scalability of XL-MIMO \cite{Villalonga-2022}. Hence, there is a strong desire to find new degree of freedom (DoF) in the physical layer. A natural direction is to deploy more antennas at user equipment (UE) but the limited physical space restricts the number of antennas, as the antenna spacing should be at least half a wavelength \cite{phone}. In a nutshell, MIMO alone may be insufficient and a new DoF will need to be sought \cite{6G}. 

To address these challenges, a new form of reconfigurable antennas advocating shape and position flexibility, known as fluid antenna system (FAS) \cite{Wong-ell2020,FAS_Lee,New-submit2024}, has emerged. The concept is greatly motivated by the recent advances in antenna technologies such as liquid metals like mercury and gallium-based alloys \cite{liquid_antenna1,liquid_antenna2,liquid_antenna3}, non-metallic conductive liquids  \cite{liquid_antenna4,liquid_antenna5,liquid_antenna}, movable antennas using stepper motors \cite{liquid_antenna6,liquid_antenna7}, holographic meta-material antennas \cite{liquid_antenna8,Hoang-2021,Deng-2023}, and pixel reconfigurable antennas \cite{pixel_antenna,Jing-2022}. For liquid antennas and movable antennas, physical limitations on the acceleration and velocity of the structures may restrict their ability to reconfigure at high speeds. In contrast, holographic meta-material antennas and pixel reconfigurable antennas are particularly promising due to their rapid response times, making them more suitable for high-speed reconfiguration. Unlike traditional fixed-position antennas, FAS offers the advantage of not being fixed at a specific location, empowering them to switch to more advantageous positions as needed, providing great diversity and reconfigurability. Proof-of-concepts using fluidic materials and reconfigurable pixels have recently been reported in \cite{Shen-tap_submit2024} and \cite{Zhang-pFAS2024}, respectively.

Back in 2020, FAS was first introduced to wireless communications by Wong {\em et al.}~in \cite{FAS_ER,FAS}. Since then, there have been efforts to improve the modelling accuracy of the spatial correlation among the FAS ports \cite{wong2022closed,Khammassi-2023,ramirez2024new}. While FAS can be deployed at the transmitter and/or receiver side(s), for single-user channels, most studies focus on using FAS at the receiver since channel state information (CSI) is only required at the receiver. According to the nomenclature in \cite{New-submit2024}, this is referred to as the Rx-SISO-FAS model and its diversity has been studied in \cite{FAS_close} considering the model in \cite{Khammassi-2023}. Recently in \cite{New-twc2023}, the dual-MIMO-FAS model which considered multiple fluid antennas at both ends was investigated and the diversity-multiplexing trade-off of that model was derived. The results for FAS have been promising, demonstrating great potential even with a small physical size of FAS. However, previous results appeared to focus primarily on Rayleigh fading channels where only non-line-of-sight (NLoS) component was present. In \cite{FAS_nakagami}, a more general Nakagami-$m$ fading channel model was considered for FAS, but the outage probability (OP) is provided in open form, while closed-form approximation and the ergodic rate (ER) are not included. When the fading parameter of Nakagami-m fading $m = \frac{(\kappa+1)^2}{2\kappa+1}$, the distribution of Nakagami-m is approximately Rician fading with parameter $\kappa$ \cite{wireless}. Also, only a uniform linear array (ULA) configuration for the FAS ports, i.e., the ports are evenly distributed in a linear space, was considered. How the port configuration of FAS, ULA or uniform planar array (UPA), would affect the outage performance is not known.

Motivated by this, in this paper, we analyze the performance of the Rx-SISO-FAS model under spatially-correlated Rician fading channels considering both ULA and UPA port configurations. As in previous studies, we will assume that the CSI is perfectly known. In practice, the CSI can be estimated using the sparse signal processing methods in \cite{Hao-2024,Dai-2023}. 

Our main contributions are summarized as follows:
\begin{itemize}
\item We consider the Rx-SISO-FAS model under Rician fading channels where a position-flexible antenna can switch freely among \(N\) ports at the receiver to have the strongest channel gain. Under this model, we derive new channel statistics such as the joint probability density function (PDF) and cumulative distribution function (CDF) of the channel gain, accounting for the port correlation.
\item Additionally, we obtain accurate expressions for the OP based on the PDF and CDF. Closed-form expressions for upper and lower bounds are subsequently obtained. We then develop exact solutions and closed-form approximations for the ER. The diversity order is derived based on the upper bound of OP, which is shown to be only related to the number of ports, $N$, when the signal-to-noise ratio (SNR) is high enough.
\item Then we consider both ULA and UPA models for FAS when deriving expressions for the OP and ER. Through numerical results, we compare the performance of ULA and UPA, showing that the UPA port configuration can improve the outage performance and increase the ER under different Rician factors.
\item Our simulation results demonstrate that 1) FAS remains an effective technique under Rician fading channels; 2) a strong line-of-sight (LoS) component negatively impacts the performance of FAS at low SNR; 3) FAS can outperform an $L$ branches maximum ratio combining (MRC) system under Rician fading channels; and 4) the UPA port configuration outperforms the ULA counterpart.
\end{itemize}

\subsection{Organization and Notations}
The remainder of this paper is organized as follows. Section \ref{sec:model} introduces the channel model of Rx-SISO-FAS with ULA port configuration. Our analytical results for the Rx-SISO-FAS model with ULA are presented in Section \ref{sec:analysis}. In Section \ref{sec:upa}, we turn our attention to the UPA port configuration and derive new analytical results. Section \ref{sec:results} provides the numerical results and finally, Section \ref{sec:conclude} concludes this paper. 

The following notation will be used throughout this paper. The symbol \( |\cdot| \) represents the absolute value, and \( \mathbb{E}(\cdot) \) denotes the expectation operator. The expression \( X|Y \) signifies that \( X \) is conditioned on \( Y \). Additionally, \( X \sim \mathcal{N}_c (\mu, \sigma^2) \) indicates that the random variable \( X \) follows a complex Gaussian distribution with mean \( \mu \) and variance \( \sigma^2 \). In the absence of the subscript \( c \), it is assumed that \( X \) is a real Gaussian random variable.

\begin{figure}[]
\centering
\includegraphics[width=3.4in]{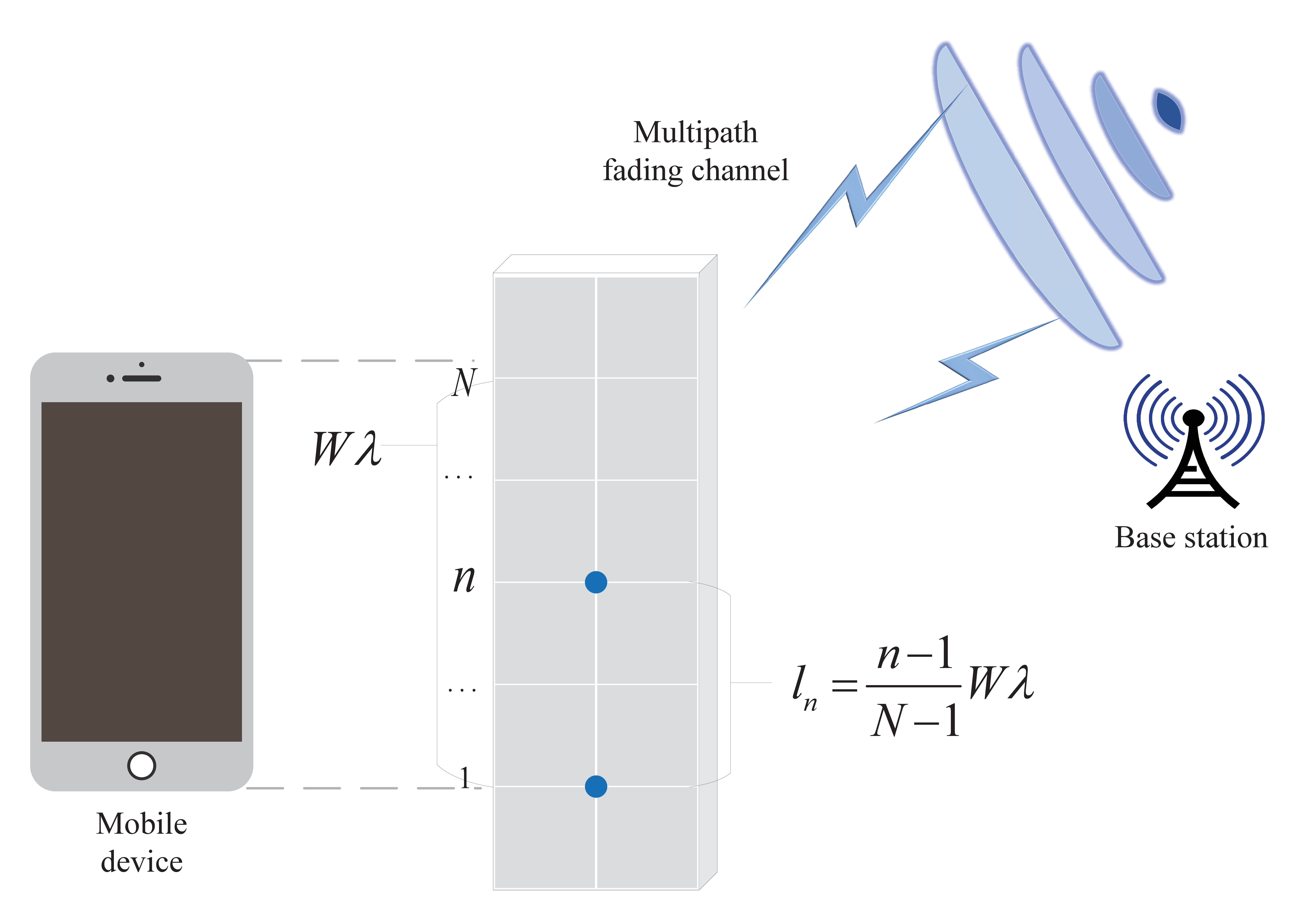}
\caption{The ULA port configuration for FAS.}\label{system1}
\vspace{-3mm}
\end{figure}

\section{System and Channel Model}\label{sec:model}
Here, we consider a single-input single-output (SISO) system where the transmitter is equipped with a conventional fixed-position antenna but the receiver has a FAS with $N$ ports (or flexible positions) linearly distributed along one dimension, as depicted in Fig.~\ref{system1}. This is classified as the Rx-SISO-FAS model according to \cite{New-submit2024} and this port arrangement is referred to as the ULA port configuration. The beauty of FAS is that its reception position can be switched between different ports for optimizing communication performance. 

Assuming equal spacing between two adjacent ports and using the first port as a distance reference, the displacement of the $n$-th port is given by
\begin{equation}\label{l}
{l_{n,1}} = \left(\frac{{n - 1}}{{N - 1}}\right)W\lambda ,~\mbox{for }n = 1,2,\dots,N,
\end{equation}
where $W\lambda$ is the total physical length of FAS, $W$ is an integer, $\lambda$ denotes the carrier wavelength, and $N$ is the number of ports.

The received signal is subjected to an additive white Gaussian noise (AWGN), so that the received signal at the $n$-th port is written as
\begin{equation}\label{AWGN}
{ y_n^{FAS}} = l\left( d \right){h_n}ps + {\eta _n},
\end{equation}
in which $l\left( d \right)$ is the large-scale fading coefficient, $h_n \sim \mathcal{N}_c (A, \sigma^2)$ is the small-scale fading coefficient at the $n$-th port, $p$ is the transmit power of the BS, $\eta_n \sim \mathcal{N}_c (0, \sigma_\eta^2)$ is the noise, and $ s $ represents the transmitted signal with $\mathbb{E}(s)=0$. The large-scale fading coefficient between the BS and mobile device is defined as:
\begin{equation}\label{The large-scale fading channel}
l\left( d \right) = {d^{ - 2 }}{\left( {\frac{\lambda}{{4\pi}}} \right)^2},
\end{equation}
where ${d}$ denotes the distance between the BS and mobile device. 

Since the classic Rayleigh fading channel is limited to modeling strong scattering scenarios, it is not a viable solution for future communication systems. To account for the presence of LoS links between the BS and UE, the PDF of small-scale fading is characterized using the Rician fading channel.The amplitude of the channel, $|h_n|$, follows Rician distribution, and the PDF is given by
\begin{equation}\label{rician dis}
{p_{\left| {{h_{{n}}}} \right|}}(m) = \frac{{2m}}{{{\sigma ^2}}}{e^{ - \frac{{{m^2} + {A^2}}}{{{\sigma ^2}}}}}{I_0}\left( {\frac{{2mA}}{{{\sigma ^2}}}} \right),
\end{equation}
where $I_0(\cdot)$ is the zero-order modified Bessel function of the first kind, and  $\kappa$ denotes the Rician fading factor defined as
\begin{equation}\label{rician factor}
\kappa = \frac{{{A^2}}}{{{\sigma ^2}}},
\end{equation}
where $A^2$ and $\sigma ^2$ can be, respectively, interpreted as the power of LoS path and NLoS path. According to the PDF in (\ref{rician dis}) and the requirement of power normalization,  we can obtain
\begin{equation}\label{rician factor}
\begin{aligned}
\mathbb{E}(|h_n|^2)&=\sigma^2+A^2=1,\\
{A^2} &= \frac{\kappa}{{\kappa + 1}},\\
{\sigma ^2} &= \frac{1}{{\kappa + 1}}.
\end{aligned}
\end{equation}

The received SNR at the $n$-th port can be found as
\begin{equation}\label{SNR}
{\rm SNR}_n = \frac{{l\left( d \right){{\left| {{h_n}} \right|}^2}p{\mathbb{E}}\left( {{{\left| s \right|}^2}} \right)}}{{\sigma _\eta ^2{\rm{ }}}}.
\end{equation}
Hence, the average received SNR can be obtained as
\begin{equation}\label{ESNR}
\begin{aligned}
\overline {\rm SNR}_{n}  &= \frac{l\left( d \right){\mathbb E\left( {{{\left| {{h_n}} \right|}^2}} \right)p{\mathbb E}\left( {{{\left| s \right|}^2}} \right)}}{{\sigma _\eta ^2{\rm{ }}}} \\
&= \frac{{l\left( d \right)\left( {{A^2} + {\sigma ^2}} \right)p{\mathbb E}\left( {{{\left| s \right|}^2}} \right)}}{{\sigma _\eta ^2{\rm{ }}}}.
\end{aligned}
\end{equation}

The different port locations make sure that the multipath arriving have phase differences, resulting in spatial correlation between the channels at different ports. Under rich scattering, the correlation factor can be modelled as \cite{FAS_ER,FAS}
\begin{equation}\label{rou}
\begin{aligned}
{\rho _n} &= {J_0}\left( {2\pi \frac{{{l_{n,1}}}}{\lambda }} \right)\\
& = {J_0}\left( {2\pi \frac{{\left( {n - 1} \right)W\lambda }}{{\left( {N - 1} \right)\lambda }}} \right)\\
 &= {J_0}\left( {\frac{{2\pi (n - 1)}}{{N - 1}}W} \right),\;{\rm{for}}\;n = 1,2, \ldots ,N,
\end{aligned}
\end{equation}
where $J_0(\cdot)$ is the zero-order Bessel function of the first kind. With that, the channels at the ports can be parameterized as
\begin{equation}\label{channel model-0}
\left\{ \begin{aligned}
{h_1} &= \sigma {x_0} + j\sigma {y_0} + A\\
{h_n} &= \sigma \left( {\sqrt {1 - \rho _n^2} {x_n} + {\rho _n}{x_0}} \right) \\
&+ j\sigma \left( {\sqrt {1 - \rho _n^2} {y_n} 
+ {\rho _n}{y_0}} \right) + A\\
&\quad\quad\mbox{for }n=2,\dots,N,
\end{aligned}\right.
\end{equation}
where $x_0,y_0,x_n,y_n \sim \mathcal{N}(0, 0.5)$ and that they are all independent random variables.

\section{Performance Analysis}\label{sec:analysis}
Here, we analyze the performance of FAS by providing the expressions for OP and ER. With position reconfigurability of FAS, it is assumed that the port with the strongest channel gain is always selected. Therefore, the resultant channel gain of FAS, $\left| {{h_{\rm FAS}}} \right|$, is given by
\begin{equation}\label{channel model}
\left| {{h_{\rm FAS}}} \right| = \max \left\{ {\left| {{h_1}} \right|,\left| {{h_2}} \right|,\dots,\left| {{h_N}} \right|} \right\}.
\end{equation}
Before that, we find it important to first derive the joint distribution of ${\left| {{h_1}} \right|,|h_2|,\dots,\left| {{h_N}} \right|}$.

\subsection{PDF and CDF}
\begin{lemma}\label{Lemma1:joint pdf}
With the small-scale fading following Rician distribution, the joint PDF of the channel gains over all the ports of Rx-SISO-FAS, $|h_1|,|h_2|,\dots,|h_N|$, is given by
\begin{multline}\label{jpdf}
{P_{\left| {{h_1}} \right|,\left| {{h_2}} \right|,\dots,\left| {{h_N}} \right|}}({m_1},{m_2},\dots,{m_N})\\
= \prod_{n = 1}^N {\frac{{2{m_n}}}{{{\sigma ^2}(1 - \rho _n^2)}}} {e^{ - \frac{{m_n^2 + \rho _n^2m_1^2 + (1 - \rho _n^2){A^2}}}{{{\sigma ^2}(1 - \rho _n^2)}}}}\\
\times {I_0}\left( {\frac{{2\left( {\sqrt {\rho _n^2m_1^2 + (1 - \rho _n^2){A^2}} {m_n}} \right)}}{{{\sigma ^2}(1 - \rho _n^2)}}} \right),
\end{multline}
where $\rho _1\triangleq 0$.
\end{lemma}

\begin{proof}
See Appendix A.
\end{proof}

\begin{lemma}\label{Lemma2:joint CDF}
With the small-scale fading following Rician distribution with the Rician fading factor, $\kappa$, the joint CDF of the channel gains over all the ports of Rx-SISO-FAS, $|h_1|,|h_2|,\dots,|h_N|$, is given by
\begin{equation}\label{jcdf}
\begin{aligned}
& {F_{\left| {{h_1}} \right|,\left| {{h_2}} \right|,\dots,\left| {{h_N}} \right|}} \left({m_1},{m_2},\dots,{m_N}\right)\\
&= P(\left| {{h_1}} \right| < {m_1},\left| {{h_2}} \right| < {m_2},\dots,\left| {{h_N}} \right| < {m_N})\\
&= \int_0^{{m_1}} {\frac{{2{t_1}}}{{{\sigma ^2}}}{e^{ - \frac{{t_1^2 + {A^2}}}{{{\sigma ^2}}}}}{I_0}\left( {\frac{{2{t_1}A}}{{{\sigma ^2}}}} \right)} \times\\
&\prod\limits_{n = 2}^N {\left( {1 - {Q_1}\left[ {\sqrt {\frac{{2\rho _n^2t_1^2}}{{{\sigma ^2}\left( {1 - \rho _n^2} \right)}} + 2\kappa} ,\sqrt {\frac{{2(1 + \kappa)}}{{{\sigma ^2}(1 - \rho _n^2)}}} {m_n}} \right]} \right)} d{t_1},
\end{aligned}
\end{equation}
where $Q_1( \cdot )$ is the first-order Marcum Q-function \cite{marcumq}.
\end{lemma}

\begin{proof}
See Appendix B.
\end{proof}

\subsection{OP and ER}
We define the outage event as
\begin{equation}\label{NSNR}
\left\{ {\frac{{\rm SNR}_{\rm FAS}}{{\overline {\rm SNR} }_n} = \frac{{{{\left| {{h_{\rm FAS}}} \right|}^2}}}{{{A^2} + {\sigma ^2}}} < {\gamma _{th}}} \right\},
\end{equation}
where $\gamma _{th}$ represents the SNR threshold for the normalized SNR. With the joint CDF in \textbf{Lemma~\ref{Lemma2:joint CDF}}, we can compute the OP of the proposed FAS. The outage happens when the power of the received signal is lower than the required threshold.

\begin{theorem}\label{Theorem3:outage probability}
With the definition of outage event in~\eqref{NSNR}, the OP is defined as the probability that the channel gain of FAS is lower than the required threshold. In what follows, the OP of the Rx-SISO-FAS model with ULA port configuration under Rician fading channels can be obtained by
\begin{equation}\label{pout}
\begin{aligned}
&{p_{out}}\left( {{\gamma _{th}}} \right) = \int_0^{{\gamma _{th}}} {{e^{ - \left( {\left( {\kappa + 1} \right)t + \kappa} \right)}}{I_0}} \left( {2\sqrt {\kappa(\kappa + 1)t} } \right)\\
& \prod\limits_{n = 2}^N {\left[ {1 - {Q_1}\left( {\sqrt {\frac{{2\rho _n^2(\kappa + 1)t}}{{(1 - \rho _n^2)}} + 2\kappa} ,\sqrt {\frac{{2(1 + \kappa)}}{{(1 - \rho _n^2)}}} \sqrt {{\gamma _{th}}} } \right)} \right]} dt.
\end{aligned}
\end{equation}
\end{theorem}

\begin{proof}
The OP is found by substituting \(m_1=\dots=m_N=\gamma_{th}\)  into the joint CDF in~\eqref{jcdf}, which completes the proof.
\end{proof}

\begin{corollary}\label{Corollary1:lower bound}
A lower bound of the OP with Rician fading factor $\kappa$ and the SNR threshold $\gamma_{th}$ can be found as
\begin{equation}\label{out-up}
\begin{aligned}
&{p_{out}}\left( {{\gamma _{th}}} \right)\ge{\hat p_{out}}\left( {{\gamma _{th}}} \right) \\
&= \left[ {1 - {Q_1}\left( {\sqrt { 2\kappa} ,\sqrt {2(1 + \kappa)} \sqrt {{\gamma _{th}}} } \right)} \right]\times\\
&\prod\limits_{n = 2}^N {\left(1-{Q_1}\left( {\sqrt {\frac{{2\rho _n^2(\kappa + 1)\gamma _{th}}}{{(1 - \rho _n^2)}} + 2\kappa} ,\sqrt {\frac{{2(1 + \kappa)}}{{(1 - \rho _n^2)}}} \sqrt {{\gamma _{th}}} } \right)\right)} .
\end{aligned}
\end{equation}
\end{corollary}

\begin{proof}
See Appendix C.
\end{proof}

\begin{corollary}\label{outup}
An upper bound of the OP with Rician fading factor $\kappa$ and the SNR threshold $\gamma_{th}$ is given by
\begin{equation}\label{out-low}
\begin{aligned}
{p_{out}}\left( {{\gamma _{th}}} \right)&\le{{\tilde p_{out}}}\left( {{\gamma _{th}}} \right) \\
&= \left[ {1 - {Q_1}\left( {\sqrt {2\kappa} ,\sqrt {2(1 + \kappa)} \sqrt {{\gamma _{th}}} } \right)} \right]\\
&\times\prod\limits_{n = 2}^N {\left( {1 - {\alpha _n}{e^{ - \frac{c }{{1 - \rho _n^2}}{\gamma _\kappa}}}} \right)} ,
\end{aligned}
\end{equation}
where 
\begin{align*}
{\gamma _\kappa} &= {\gamma _{th}}\left( {\kappa + 1} \right) + \kappa - 2\sqrt {{\gamma _{th}}\kappa(\kappa + 1)},\\
{\alpha _n} &= \frac{{\alpha {{\left[ {{\gamma _{th}}\left( {1 + \kappa} \right)} \right]}^{0.25}}}}{{\sqrt {\left| {{\rho _n}} \right|} {{\left[ {{\gamma _{th}}\left( {1 + \kappa} \right)} \right]}^{0.25}} + {{\left[ {\kappa\left( {1 - \rho _n^2} \right)} \right]}^{0.25}}}},\\
\alpha &= \frac{{{e^{\frac{1}{{\left[ {\pi \left( {c  - 1} \right) + 2} \right]}}}}}}{{2c }}\sqrt {\frac{1}{\pi }\left( {c  - 1} \right)\left[ {\pi \left( {c  - 1} \right) + 2} \right]},
\end{align*}
and $c$ is a constant number greater than one.
\end{corollary}

\begin{proof}
See Appendix D.
\end{proof}

\begin{corollary}\label{Corollary3:lower bound in series}
The lower bound of the OP in \eqref{out-up} can be represented in series as:
\begin{equation}\label{out-up-series}
{{\hat p}_{out}} = \prod\limits_{n = 1}^N {\left( {{e^{ - a_n^2/2}}\sum\limits_{k = 0}^\infty  {\frac{1}{{k!}}} \frac{{\gamma \left( {1 + k,\frac{{b_n^2}}{2}} \right)}}{{\Gamma \left( {1 + k} \right)}}{{\left( {\frac{{a_n^2}}{2}} \right)}^k}} \right)} 
\end{equation}
where 
\begin{align*}
{a_n} &= \sqrt {\frac{{2\rho _n^2(\kappa  + 1){\gamma _{th}}}}{{(1 - \rho _n^2)}} + 2\kappa } ,n = 1...N,\\
{b_n} &= \sqrt {\frac{{2(1 + \kappa )}}{{(1 - \rho _n^2)}}} \sqrt {{\gamma _{th}}} ,n = 1...N,
\end{align*}
\end{corollary}
and $\rho_1$ is set to $0$.
\begin{proof}
The generalized Marcum Q function of order $v>0$ can be represented using incomplete Gamma function as \cite{marcum_appro}
\begin{equation}\label{q_gamma}
Q_v\left(a,b\right)=1 - {e^{ - {a^2}/2}}\sum\limits_{k = 0}^\infty  {\frac{1}{{k!}}} \frac{{\gamma \left( {v + k,\frac{{{b^2}}}{2}} \right)}}{{\Gamma \left( {v + k} \right)}}{\left( {\frac{{{a^2}}}{2}} \right)^k}.
\end{equation}
The lower bound of OP in \eqref{out-up-series} is found by substituting \eqref{q_gamma} and $v=1$ into \eqref{out-up}, which completes the proof.
\end{proof}

\begin{corollary}\label{Corollary3:lower bound in raylei}
When the Rician fading factor $\kappa=0$, the lower bound of the OP in \eqref{out-up} can be represented as:
\begin{equation}\label{out-up-ray0}
\begin{aligned}
&{{\hat p}_{out}^{Rayleigh}}\left( {{\gamma _{th}}} \right) \\
&= \left( {1 - {e^{ - {\gamma _{th}}}}} \right)\\
 &\times \prod\limits_{n = 2}^N {\left( {1 - {Q_1}\left( {\sqrt {\frac{{2\rho _n^2{\gamma _{th}}}}{{(1 - \rho _n^2)}}} ,\sqrt {\frac{{2{\gamma _{th}}}}{{(1 - \rho _n^2)}}} } \right)} \right)} \\
 &= \left( {1 - {e^{ - {\gamma _{th}}}}} \right)\\
& \times \prod\limits_{n = 2}^N {\left( {{e^{ - {a'}_n^2/2}}\sum\limits_{k = 0}^\infty  {\frac{1}{{k!}}} \frac{{\gamma \left( {v + k,\frac{{{b'}_n^2}}{2}} \right)}}{{\Gamma \left( {v + k} \right)}}{{\left( {\frac{{{a'}_n^2}}{2}} \right)}^k}} \right)} ,
\end{aligned}
\end{equation}
where 
\begin{align*}
{{a'}_n} = \sqrt {\frac{{2\rho _n^2{\gamma _{th}}}}{{(1 - \rho _n^2)}}} ,n = 2...N,\\
{{b'}_n} = \sqrt {\frac{{2{\gamma _{th}}}}{{(1 - \rho _n^2)}}} ,n = 2...N.
\end{align*}
\end{corollary}
\begin{proof}
When the Rician fading factor \( \kappa = 0 \), the channel transitions to a Rayleigh fading channel. By substituting \( \kappa = 0 \) and $N=1$ into, the PDF of the channel amplitude simplifies to:
\begin{equation}\label{raylei dis}
{p_{\left| {{h_{{n}}}} \right|}^{Rayleigh}}(m) = \frac{{2m}}{{{\sigma ^2}}}{e^{ - \frac{{{m^2}}}{{{\sigma ^2}}}}}.
\end{equation}
Based on~\eqref{raylei dis}, we can obtain the OP when the $\kappa=0$ and $N=1$ as
\begin{equation}\label{pout_Ray}
{{p}_{out\left( 1 \right)}^{Rayleigh}} = \left( {1 - {e^{ - {\gamma _{th}}}}} \right).
\end{equation}
By substituting~\eqref{pout_Ray},~\eqref{q_gamma} and $\kappa=0$ into~\eqref{out-up}, we can obtain the desired result in~\eqref{out-up-ray0}, which completes the proof.
\end{proof}

To gain insights into the system performance, the slope of OP when the number of ports $N$ is high, is worth investigating. We first express the high-$N$ slope as
\begin{equation}\label{slopeN}
{S_\infty } =  - \mathop {\lim }\limits_{N \to \infty } \frac{{\log{p_{out}}\left( {{\gamma _{th}}} \right)}}{{\log N}}.
\end{equation}

\begin{proposition}\label{proposition2: slope}
\emph{Based on \textbf{Corollary~\ref{outup}}, when the number of ports $N$ is high, the slope of OP is given by}
\begin{equation}\label{slope}
{S_\infty } = \infty .
\end{equation}

\begin{proof}

Due to the complexity of the exact expression for OP, direct analysis is challenging. Therefore, we focus on analyzing the slope of the upper bound of the OP. If the slope of the upper bound approaches infinity when the number of ports $N$ is high, the slope of the OP will also tend towards infinity. First, we further simplify the upper bound of the OP in \textbf{Corollary~\ref{outup}} as
\begin{equation}\label{upp_app}
\begin{aligned}
{{\tilde p}_{out}}\left( {{\gamma _{th}}} \right) &= \left[ {1 - {Q_1}\left( {\sqrt {2\kappa } ,\sqrt {2(1 + \kappa )} \sqrt {{\gamma _{th}}} } \right)} \right]\\
& \times \prod\limits_{n = 2}^N {\left( {1 - {\alpha _n}{e^{ - \frac{c}{{1 - \rho _n^2}}{\gamma _\kappa }}}} \right)} \\
 &\le {\left( {1 - \alpha {e^{ - \frac{c}{{1 - {\rho ^2}}}{\gamma _\kappa }}}} \right)^N},
\end{aligned}
\end{equation}
where 
\begin{align*}
 \left| \rho  \right| &  = \max \left\{ {\left| {{\rho _2}} \right|,\left| {{\rho _3}} \right|,...,\left| {{\rho _N}} \right|} \right\}, \\
 \alpha  & = \frac{{\alpha {{\left[ {{\gamma _{th}}\left( {1 + \kappa } \right)} \right]}^{0.25}}}}{{\sqrt {\left| \rho  \right|} {{\left[ {{\gamma _{th}}\left( {1 + \kappa } \right)} \right]}^{0.25}} + {{\left[ {\kappa \left( {1 - {\rho ^2}} \right)} \right]}^{0.25}}}}.
\end{align*}

By substituting \eqref{upp_app} into \eqref{slopeN}, we have
\begin{equation}\label{sp1}
\begin{aligned}
{S_\infty } &>  - \mathop {\lim }\limits_{N \to \infty } \frac{{\log \left( {{{\left( {1 - \alpha {e^{ - \frac{c}{{1 - {\rho ^2}}}{\gamma _\kappa }}}} \right)}^N}} \right)}}{{\log N}}\\
 &=  - \mathop {\lim }\limits_{N \to \infty } \frac{{N\log \left( {1 - \alpha {e^{ - \frac{c}{{1 - {\rho ^2}}}{\gamma _\kappa }}}} \right)}}{{\log N}} = \infty .
\end{aligned}
\end{equation}
Then, the proof is complete.
\end{proof}
\end{proposition}

\begin{remark}\label{remarkop}
The result of~\eqref{slope} illustrates that the slope of OP of the Rx-SISO-FAS is infinite for large $N$. For a given SNR, any dimension \(W\lambda\), and any Rician factor \(\kappa\), the FAS can achieve an arbitrarily small OP if the number of ports \(N \to \infty\) and \(\left| \rho_n \right| \ne 1\).
\end{remark}

By analyzing the slope of the OP when the SNR is high, we can obtain the diversity order.

\begin{proposition}\label{proposition1: do}
\emph{Based on \textbf{Corollary~\ref{outup}}, when the SNR is high, the diversity order of the Rx-SISO-FAS model is given by}
\begin{equation}\label{diversity order of N}
{d_{N }} =  - \mathop {\lim }\limits_{\frac{1}{{{\gamma _{th}}}} \to \infty } \frac{{\log {{p}_{out}}}}{{\log \frac{1}{{{\gamma _{th}}}}}} \approx N.
\end{equation}

\begin{proof}
To derive the diversity order of the system, we can use the simplified upper bound of the OP in \eqref{upp_app}. By substituting \eqref{upp_app} into \eqref{diversity order of N}, we have
\begin{equation}
\begin{array}{*{20}{l}}
{d_N} &\ge  - \mathop {\lim }\limits_{\frac{1}{{{\gamma _{th}}}} \to \infty } \frac{{\log {{\tilde p}_{out}}}}{{\log \frac{1}{{{\gamma _{th}}}}}}\\
&\ge  - \mathop {\lim }\limits_{\frac{1}{{{\gamma _{th}}}} \to \infty } \frac{{\log \left( {{{\left( {1 - \alpha {e^{ - \frac{c}{{1 - {\rho ^2}}}{\gamma _\kappa }}}} \right)}^N}} \right)}}{{\log \frac{1}{{{\gamma _{th}}}}}}\\
 &\approx  - \mathop {\lim }\limits_{\frac{1}{{{\gamma _{th}}}} \to \infty } \frac{{N\log \left( {1 - \left( {1 - \frac{{c{\gamma _{th}}}}{{1 - {\rho ^2}}}} \right)} \right)}}{{\log \frac{1}{{{\gamma _{th}}}}}}\\
& \approx  - \mathop {\lim }\limits_{\frac{1}{{{\gamma _{th}}}} \to \infty } \frac{{N\log \left( {{\gamma _{th}}} \right)}}{{\log \frac{1}{{{\gamma _{th}}}}}} = N
\end{array}
\end{equation}
Then, the proof is complete.
\end{proof}
\end{proposition}

\begin{remark}\label{remarkdo}
The result of~\eqref{diversity order of N} illustrates that the diversity order of the Rx-SISO-FAS model can be approximated by the number of ports $N$. Additionally, it shows that the outage performance can be enhanced by increasing the number of ports.
\end{remark}

To further investigate the impact of the Rician fading factor on the OP, it is worthwhile estimating the slope of OP when the Rician fading factor is large. Therefore, we first express the large-\(\kappa\) slope as
\begin{equation}\label{slopek}
{\hat S_\infty } =  - \mathop {\lim }\limits_{\kappa \to \infty } \frac{{\log{P_{out}}\left( {{\gamma _{th}}} \right)}}{{\log \kappa}}.
\end{equation}

\begin{proposition}\label{proposition3: slope2}
\emph{Based on \textbf{Corollary~\ref{Corollary1:lower bound}}, when the Rician fading factor is large enough, the slope of OP is given by}
\begin{equation}\label{slope2}
{\hat S_\infty } = 0 .
\end{equation}
\begin{proof}
First, we further simplify the lower bound of the OP in \textbf{Corollary~\ref{Corollary1:lower bound}} as
\begin{equation}\label{oplowe}
\begin{aligned}
&{p_{out}}\left( {{\gamma _{th}}} \right) \ge {{\hat p}_{out}}\left( {{\gamma _{th}}} \right)\\
 &= \left[ {1 - {Q_1}\left( {\sqrt {2\kappa } ,\sqrt {2(1 + \kappa )} \sqrt {{\gamma _{th}}} } \right)} \right]\times\\
 & \prod\limits_{n = 2}^N {\left( {1 - {Q_1}\left( {\sqrt {\frac{{2\rho _n^2(\kappa  + 1){\gamma _{th}}}}{{(1 - \rho _n^2)}} + 2\kappa } ,\sqrt {\frac{{2(1 + \kappa )}}{{(1 - \rho _n^2)}}} \sqrt {{\gamma _{th}}} } \right)} \right)} \\
 &> {\left( {1 - {Q_1}\left( {\sqrt {2\kappa } ,\sqrt {2(1 + \kappa )} \sqrt {{\gamma _{th}}} } \right)} \right)^N}.
\end{aligned}
\end{equation}
In the case of $a<b$, the exponential-type upper bound of Marcum Q-function is
\begin{equation}\label{Qupp}
{Q_1}\left( {a,b} \right) \le {e^{ - \frac{{{{\left( {b - a} \right)}^2}}}{2}}}.
\end{equation}
By substituting \eqref{oplowe} and \eqref{Qupp} into \eqref{slopek}, we have
\begin{equation}
\begin{aligned}
{S_\infty } &<  - \mathop {\lim }\limits_{\kappa  \to \infty } \frac{{\log \left( {{{\left( {1 - {Q_1}\left( {\sqrt {2\kappa } ,\sqrt {2(1 + \kappa )} \sqrt {{\gamma _{th}}} } \right)} \right)}^N}} \right)}}{{\log \kappa }}\\
 &=  - \mathop {\lim }\limits_{\kappa  \to \infty } \frac{{N\log \left( {1 - {Q_1}\left( {\sqrt {2\kappa } ,\sqrt {2(1 + \kappa )} \sqrt {{\gamma _{th}}} } \right)} \right)}}{{\log \kappa }}\\
 &<  - \mathop {\lim }\limits_{\kappa  \to \infty } \frac{{N\log \left( {1 - {e^{ - \frac{{{{\left( {\sqrt {2(1 + \kappa )} \sqrt {{\gamma _{th}}}  - \sqrt {2\kappa } } \right)}^2}}}{2}}}} \right)}}{{\log \kappa }} = 0.
\end{aligned}
\end{equation}
Then, the proof is complete.
\end{proof}
\end{proposition}

\begin{remark}\label{remarkopk}
The result of~\eqref{slope2} illustrates that the slope of OP of the Rx-SISO-FAS model is $0$ for large Rician fading factors. This indicates that increasing the ratio of the LoS component is not preferred.
\end{remark}

ER is an important performance metric. The ER of the proposed FAS model can be reduced to the ER of a SISO system, due to the fact that only one channel is used for communication at any time.

\begin{theorem}\label{Theorem2:ER}
When the size of FAS, $W\lambda$, is fixed and setting $\frac{{l\left( d \right)p{\mathbb E}\left( {{{\left| s \right|}^2}} \right)}}{{\sigma _\eta ^2}} = 1$, the ER of the Rx-SISO-FAS model under the Rician fading factor $\kappa$ and the SNR threshold $\gamma_{th}$ can be obtained by
\begin{equation}\label{pe}
\begin{aligned}
{R_{N,\kappa }} &= {\rm E}\left\{ {{{\log }_2}\left( {1 + {\rm SNR}_{\rm FAS}} \right)} \right\}\\
 &=  - \int\limits_0^\infty  {{{\log }_2}\left( {1 + x} \right)} d\left( {1 - F\left( x \right)} \right)\\
 &= \frac{1}{{\ln \left( 2 \right)}}\int\limits_0^\infty  {\frac{{1 - F\left( x \right)}}{{1 + x}}dx},
\end{aligned}
\end{equation}
where 
\begin{equation}\label{cdf_snr}
\begin{aligned}
&F\left( x \right) = \int_0^x {{e^{ - \left( {\left( {\kappa  + 1} \right)t + \kappa } \right)}}{I_0}} \left( {2\sqrt {\kappa (\kappa  + 1)t} } \right)\\
&\times \prod\limits_{n = 2}^N {\left[ {1 - {Q_1}\left( {\sqrt {\frac{{2\rho _n^2(\kappa  + 1)t}}{{(1 - \rho _n^2)}} + 2\kappa } ,\sqrt {\frac{{2(1 + \kappa )}}{{(1 - \rho _n^2)}}} \sqrt x } \right)} \right]} dt
\end{aligned}
\end{equation}
 is the CDF of the resultant SNR of FAS.
\end{theorem}

\begin{proof}
According to the definition of SNR in~\eqref{SNR}, by setting $\frac{{l\left( d \right)p{\mathbb E}\left( {{{\left| s \right|}^2}} \right)}}{{\sigma _\eta ^2}} = 1$, the SNR distribution becomes the distribution of the channel gain. By performing a variable substitution in~\eqref{jcdf}, we can obtain the desired result in~\eqref{cdf_snr}, which completes the proof.
\end{proof}

\begin{corollary}\label{erlow}
Setting $\frac{{l\left( d \right)p{\mathbb E}\left( {{{\left| s \right|}^2}} \right)}}{{\sigma _\eta ^2}} = 1$, a lower bound of the ER under the Rician fading factor $\kappa$ and the SNR threshold $\gamma_{th}$ of the Rx-SISO-FAS model can be found in closed form as
\begin{equation}\label{ERclose}
\begin{aligned}
&{\hat R_{N,\kappa }} \\
&= \frac{1}{{\ln \left( 2 \right)}}\int\limits_0^\infty  {\frac{{1 -\hat F\left( x \right)}}{{1 + x}}dx} \\
& \ge \frac{1}{{\ln \left( 2 \right)}}\sum\limits_{n = 2}^N {\sum_{s \subseteq \left\{ {1,\dots,N} \right\}\atop \left| s \right| = n} {{{\left( { - 1} \right)}^{n + 1}}\prod\limits_{i \subseteq s} {{a_i}} } } \int\limits_0^\infty  {\frac{{{e^{ - \sum_{i \subseteq s} {{b_i}x} }}}}{{1 + x}}dx} \\
 &= \frac{1}{{\ln \left( 2 \right)}}\sum\limits_{n = 2}^N {\sum_{s \subseteq \left\{ {1,\dots,N} \right\}\atop \left| s \right| = n} {{{\left( { - 1} \right)}^{n + 1}}\prod\limits_{i \subseteq s} {{a_i}{e^{ - \sum\limits_{i \subseteq s} {{b_i}} }}} } } Ei\left(- {\sum\limits_{i \subseteq s} {{b_i}} } \right),
\end{aligned}
\end{equation}
where
\begin{align*}
 {b_i} &= \frac{c}{{1 - \rho _i^2}}{(\kappa+1)},\\
  {a_i} &= \frac{{{e^{\frac{1}{{\left[ {\pi \left( {c - 1} \right) + 2} \right]}}{-\kappa}}}}}{{2c\sqrt {{\rho _i}} }}\sqrt {\frac{1}{\pi }\left( {c - 1} \right)\left[ {\pi \left( {c - 1} \right) + 2} \right]},
\end{align*}
and $Ei( \cdot )$ is exponential integral.
\end{corollary}

\begin{proof}
See Appendix E.
\end{proof}

\begin{corollary}\label{erup}
Setting $\frac{{l\left( d \right)p{\mathbb E}\left( {{{\left| s \right|}^2}} \right)}}{{\sigma _\eta ^2}} = 1$, an upper bound of the ER under the Rician fading factor $\kappa$ and the SNR threshold $\gamma_{th}$ of the Rx-SISO-FAS model can be found as
\begin{equation}\label{ERup}
\begin{aligned}
&{{\tilde R}_{N,\kappa }}\\
& \le \frac{1}{{\ln \left( 2 \right)}}\int\limits_0^\infty  {\frac{1}{{1 + x}}} \\
 &- \frac{1}{{1 + x}}\left[ {1 - {Q_1}\left( {\sqrt {2(\kappa + 1)\sqrt x  + 2\kappa} ,\sqrt {2(1 + \kappa)} \sqrt x } \right)} \right]\times\\
 & \prod\limits_{k = 2}^N \left[1-{{Q_1}\left( {\sqrt {\frac{{2\rho _k^2(\kappa + 1)\sqrt x }}{{(1 - \rho _k^2)}} + 2\kappa} ,\sqrt {\frac{{2(1 + \kappa)}}{{(1 - \rho _k^2)}}} \sqrt x } \right)}\right] dx.
\end{aligned}
\end{equation}
\end{corollary}

\begin{proof}
The upper bound of ER can be obtained by substituting the lower bound of OP in \eqref{out-up} into the expression of ER.
\end{proof}

\begin{figure}[]
\centering
\includegraphics[width =3.5in]{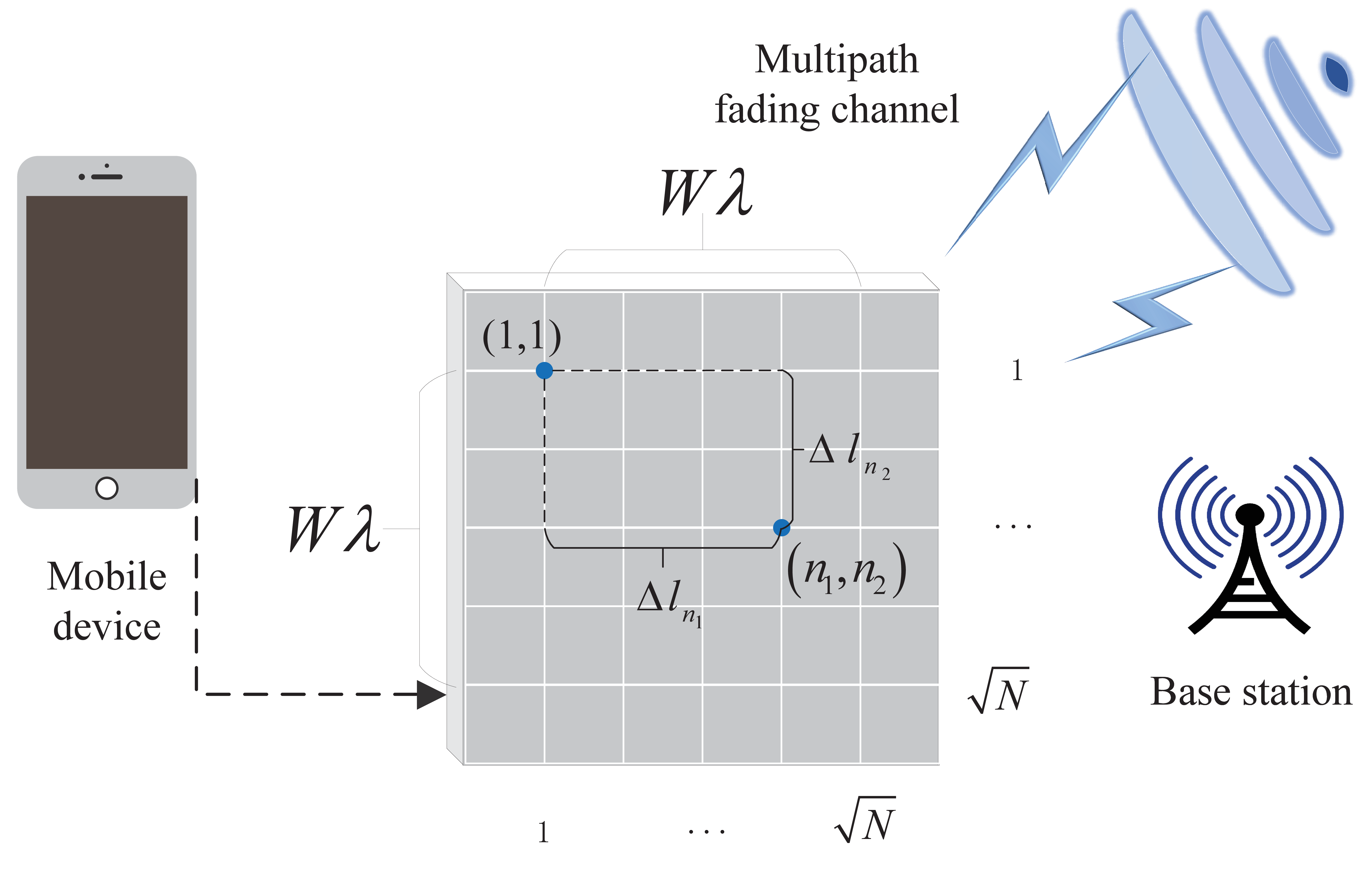}
\caption{The UPA port configuration for FAS.}\label{system2}
\end{figure}

\section{Results for UPA Port Configuration}\label{sec:upa}
In this section, we consider the case when the FAS has a UPA port configuration and extend the analytical results. As shown in Fig.~\ref{system2}, in this case, the FAS has a square planar structure. We assume that the ports are uniformly distributed over the planar surface with a size of \(W\lambda \times W\lambda\). To compare the performance with the ULA case, we set the total number of ports to \(N\), resulting in \(\sqrt{N}\) ports in each row and each column. Using point \((1,1)\) as a position reference, we have the horizontal and vertical distances to point \((n_1,n_2)\) as
\begin{equation}\label{dUPA}
\begin{aligned}
\Delta {l_{{{\rm{n}}_1}}} &= \left(\frac{{{n_1} - 1}}{{\sqrt N - 1}}\right)W\lambda,~\mbox{for }{n_1} = 1,2,\dots,\sqrt N,\\
\Delta {l_{{{\rm{n}}_2}}} &= \left(\frac{{{n_2} - 1}}{{\sqrt N - 1}}\right)W\lambda,~\mbox{for }{n_2} = 1,2,\dots,\sqrt N.
\end{aligned}
\end{equation}
This allows us to calculate the Euclidean distance from point \((n_1,n_2)\) to point \((1,1)\) as
\begin{equation}\label{pl}
\Delta {l_{{{\rm{n}}_1},{n_2}}} = \sqrt {\Delta l_{{n_1}}^2 + \Delta l_{{n_2}}^2}.
\end{equation}
For convenience, we label the ports from left top to right bottom with $\hat n$, resulting in the new expression for the correlation factor, which is given as
 \begin{equation}\label{rouUPA}
\begin{aligned}
{{\rho _ {\hat n}} } = {J_0}\left( {\frac{{2\pi \Delta {l_{{{\rm{n}}_1},{n_2}}}}}{\lambda }} \right),~\mbox{for }{n_1} &= 1,2,\dots,\sqrt N,\\
~{n_2} &= 1,2,\dots,\sqrt N,\\
 \hat n &= {n_1}\sqrt N  + {n_2}.
\end{aligned}
\end{equation}

\begin{remark}
  Based on \eqref{pl} and \eqref{l}, for the same number of ports, the distance between the $\hat n$-th port and the reference port in UPA is longer than that of the $n$-th port and the reference port in ULA. This increased distance leads to a smaller correlation factor in the UPA configuration.
\end{remark}

\begin{theorem}
the OP of the Rx-SISO-FAS model with UPA port configuration under Rician fading channels can be obtained by
\begin{equation}\label{pout_UPA}
\begin{aligned}
&{p_{out}}\left( {{\gamma _{th}}} \right) = \int_0^{{\gamma _{th}}} {{e^{ - \left( {\left( {\kappa + 1} \right)t + \kappa} \right)}}{I_0}} \left( {2\sqrt {\kappa(\kappa + 1)t} } \right)\\
& \prod\limits_{\hat n = 2}^N {\left[ {1 - {Q_1}\left( {\sqrt {\frac{{2\rho _{\hat n}^2(\kappa + 1)t}}{{(1 - \rho _{\hat n}^2)}} + 2\kappa} ,\sqrt {\frac{{2(1 + \kappa)}}{{(1 - \rho _{\hat n}^2)}}} \sqrt {{\gamma _{th}}} } \right)} \right]} dt.
\end{aligned}
\end{equation}
\begin{proof}
The OP expression can be found by substituting~\eqref{rouUPA}  into~\eqref{pout}, thereby completing the proof.
\end{proof}
\end{theorem}

\begin{corollary}
A lower bound of the OP under the SNR threshold $\gamma_{th}$ of the Rx-SISO-FAS model with UPA port configuration in Rician fading channels is obtained as
\begin{equation}\label{OPUPA1}
\begin{aligned}
&{\hat p_{out}}\left( {{\gamma _{th}}} \right) \\
&= \left[ {1 -  {Q_1}\left( {\sqrt { 2\kappa} ,\sqrt {2(1 + \kappa)} \sqrt {{\gamma _{th}}} } \right)} \right]\times\\
&\prod\limits_{\hat n = 2}^N {\left(1-{Q_1}\left( {\sqrt {\frac{{2\rho _{\hat n}^2(\kappa + 1)\gamma _{th}}}{{(1 - \rho _{\hat n}^2)}} + 2\kappa} ,\sqrt {\frac{{2(1 + \kappa)}}{{(1 - \rho _{\hat n}^2)}}} \sqrt {{\gamma _{th}}} } \right)\right)} .
\end{aligned}
\end{equation}
\end{corollary}

\begin{proof}
The OP expression can be found by substituting~\eqref{rouUPA}  into~\eqref{out-up}, thereby completing the proof.
\end{proof}

\begin{corollary}
An upper bound of the OP under Rician fading factor $\kappa$ and the SNR threshold $\gamma_{th}$ of the Rx-SISO-FAS model with UPA port configuration is found as
\begin{equation}\label{pe}
\begin{aligned}
{{\tilde p_{out}}}\left( {{\gamma _{th}}} \right) &= \left[ {1 - {Q_1}\left( {\sqrt { 2\kappa} ,\sqrt {2(1 + \kappa)} \sqrt {{\gamma _{th}}} } \right)} \right]\\
&\times \prod \limits_{\hat n = 2}^N {\left( {1 - {\alpha _{\hat n}}{e^{ - \frac{c }{{1 - \rho _{\hat n}^2}}{\gamma _\kappa}}}} \right)} ,
\end{aligned}
\end{equation}
where 
\begin{align*}
{\alpha _{\hat n}} &= \frac{{\alpha {{\left[ {{\gamma _{th}}\left( {1 + \kappa} \right)} \right]}^{0.25}}}}{{\sqrt {\left| {{\rho _{\hat n}}} \right|} {{\left[ {{\gamma _{th}}\left( {1 + \kappa} \right)} \right]}^{0.25}} + {{\left[ {\kappa\left( {1 - \rho _{\hat n}^2} \right)} \right]}^{0.25}}}}.
\end{align*}
\end{corollary}

\begin{proof}
The OP expression can be found by substituting~\eqref{rouUPA}  into~\eqref{out-low}, thereby completing the proof.
\end{proof}

\begin{corollary}\label{erUPA}
Setting $\frac{{l\left( d \right)p{\mathbb E}\left( {{{\left| s \right|}^2}} \right)}}{{\sigma _\eta ^2}} = 1$, the lower bound of the ER under the Rician fading factor $\kappa$ and the SNR threshold $\gamma_{th}$ of the Rx-SISO-FAS model can be found in closed form as
\begin{equation}\label{pe}
\begin{aligned}
{\hat R_{N,\kappa }}  &= \frac{1}{{\ln \left( 2 \right)}}\\
&\times\sum\limits_{\hat n = 2}^N {\sum_{s \subseteq \left\{ {1,2,\dots,N} \right\}\atop \left| s \right| = \hat n} {{{\left( { - 1} \right)}^{n + 1}}\prod\limits_{i \subseteq s} {{a_i}{e^{ - \sum\limits_{i \subseteq s} {{b_i}} }}} } } Ei\left( {\sum\limits_{i \subseteq s} {{b_i}} } \right),
\end{aligned}
\end{equation}
\end{corollary}

\begin{proof}
The ER expression can be found by substituting~\eqref{rouUPA}  into~\eqref{ERclose}, thereby completing the proof.
\end{proof}

\begin{figure}[t!]
\centering
\includegraphics[width =3.5in]{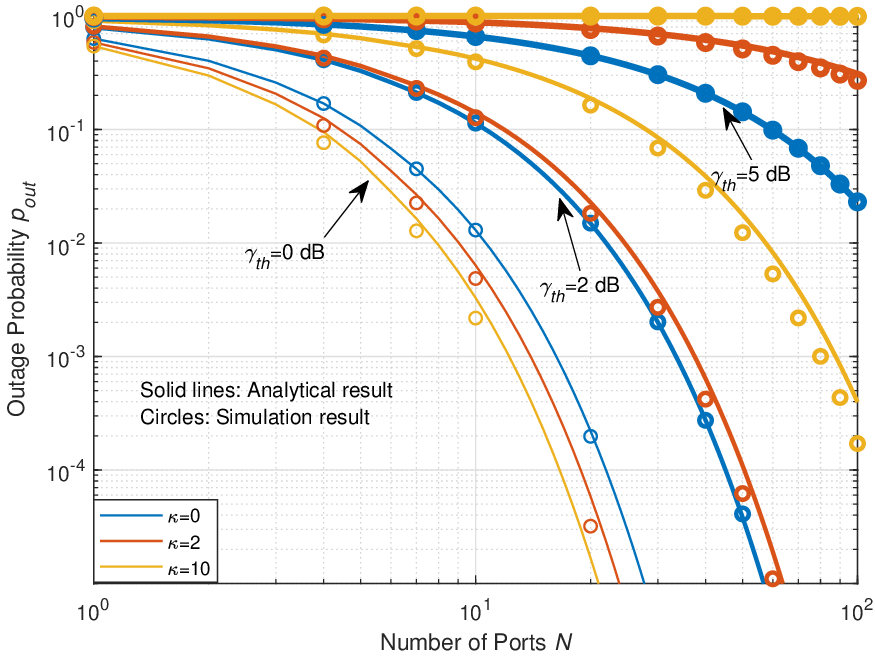}
\caption{The OP against the number of ports $N$ under different Rician fading factor $\kappa$ and SNR threshold $\gamma_{th}$, where the analytical results are derived from \eqref{pout}.}\label{opn}
\end{figure}

\section{Numerical Results}\label{sec:results}
Here we provide numerical results to study the performance of the Rx-SISO-FAS in Rician fading channels considering both ULA and UPA port configurations. The accuracy of our analytical results is verified through Monte Carlo simulations. In the simulations, unless otherwise specified, the size of FAS is set to \(W\lambda=2\lambda\). The SNR threshold $\gamma_{th}$ is expressed in dB, while all other parameters are given in linear values. We focus on the results of OP and ER. The number of Monte Carlo simulations used in this paper is $ 1 \times 10^6 $ for OP and $ 1 \times 10^2 $ for ER.

\begin{figure}[t!]
\centering
\includegraphics[width =3.5in]{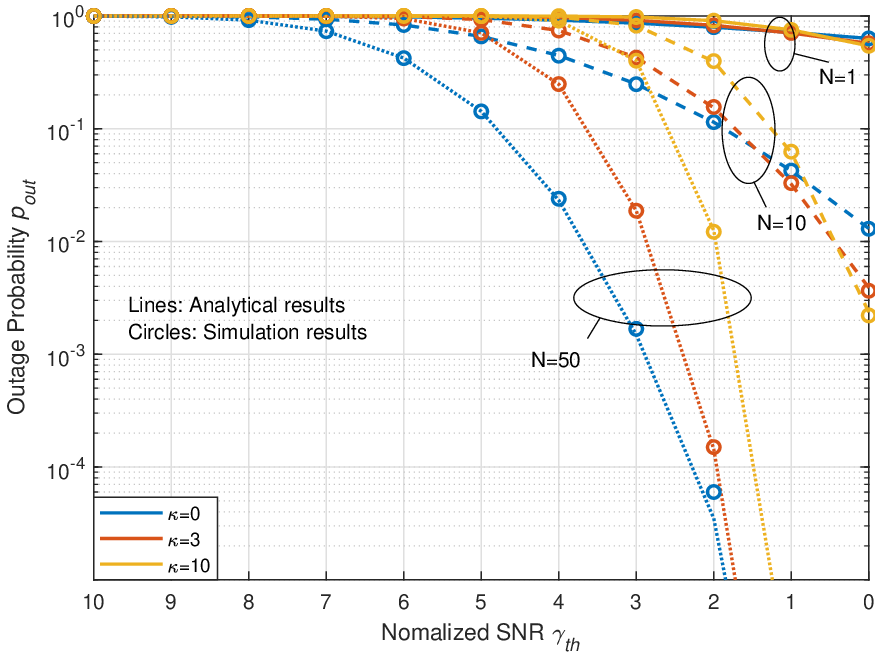}
\caption{The OP against the SNR threshold $\gamma_{th}$ under different Rician fading factor $\kappa$ and number of ports $N$.}\label{opsnr}
\end{figure}

\subsection{Impact of the Number of Ports on OP}
In Fig.~\ref{opn}, we analyze the OP performance as a function of the number of ports \(N\) under different Rician fading factors \(\kappa\) and SNR threshold $\gamma_{th}$. The solid lines and circles represent the analytical results and Monte Carlo simulations, respectively. It is clear that as the number of ports \(N\) increases, the OP continuously decreases. This is because that as more ports are available, the received signal power can be significantly increased as a benefit of the increased diversity order. As we can see, the slope of the curves increases with the number of ports $N$, which validates our \textbf{Remark~\ref{remarkop}}. However, it is noteworthy that in the high SNR regime, an increase in the Rician fading factor \(\kappa\) enhances the overall performance. Conversely, in the low SNR regime, the OP with a larger Rician fading factor $\kappa=10$ significantly deteriorates. As the Rician fading factor $\kappa$ increases, the LoS component becomes more significant. Unlike conventional wireless networks, where a strong LoS component is expected, in the Rx-SISO-FAS, strong LoS impacts negatively on the performance in the low SNR regime.

\subsection{Impact of SNR Threshold on OP} 
We turn our attention to the impact of the SNR threshold on the OP performance using the results in Fig.~\ref{opsnr}. The results illustrate that as expected, if the SNR threshold $\gamma_{th}$ increases (to the left direction), the OP increases quickly. Additionally, it can be seen that as the SNR threshold $\gamma_{th}$ decreases (to the right direction), the slope of OP gets steeper, which verifies the analysis of \textbf{Remark~\ref{remarkdo}}. It is also observed that for FAS with more ports, the slope of the curve is steeper, but the slope could be flattened in FAS with fewer ports. This indicates that increasing the transmission power can significantly enhance the outage performance of FAS but only with enough ports. In addition, a smaller Rician fading factor $\kappa$ exhibits a lower OP, especially in the low SNR regime. 

\begin{figure}[t!]
\centering
\includegraphics[width =3.5in]{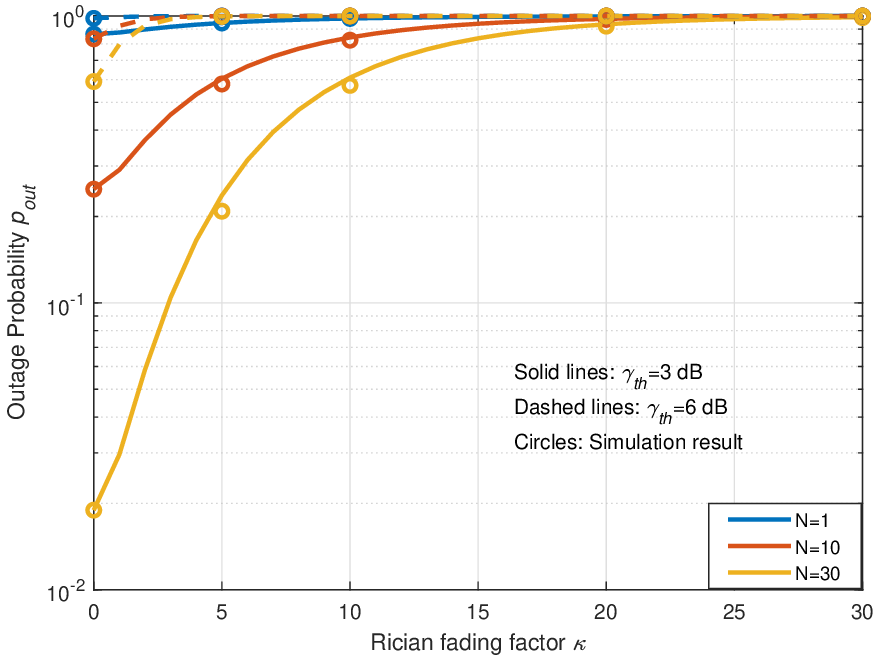}
\caption{The OP against the Rician fading factor $\kappa$ under different number of ports $N$ when $\gamma_{th}$ = 3 dB and $\gamma_{th}$ = 6 dB.}\label{opk}
\end{figure}

\subsection{Impact of Rician Factor on OP} 
Following the results in Figs.~\ref{opn} and \ref{opsnr}, we further analyze the influence of the Rician fading factor \(\kappa\) on the OP performance in Fig.~\ref{opk}. In this figure, the OP is shown against the Rician fading factor \(\kappa\) for different numbers of ports \(N\) and SNR threshold \(\gamma_{th}\). It is observed that the slope of the OP curve drops to zero as the Rician fading factor \(\kappa\) increases, verifying the analysis in \textbf{Remark~\ref{remarkopk}}. Equally, the OP rapidly deteriorates to nearly 1 with an increasing Rician fading factor \(\kappa\). This further illustrates that in the low SNR regime, a strong LoS component could actually degrade system performance. Additionally, it can be observed that increasing the number of ports \(N\) significantly improves the outage performance when \(\kappa\) is small. However, as the LoS component becomes stronger, the effectiveness of adding more ports diminishes.

\begin{figure}[t!]
\centering
\includegraphics[width =3.5in]{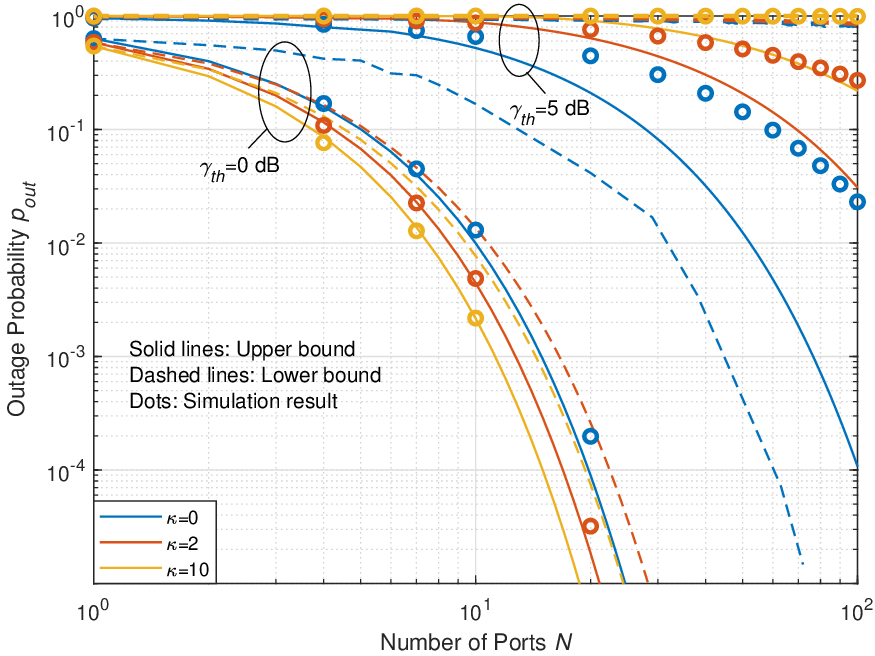}
\caption{Upper and lower bounds of OP when $\gamma_{th}$ = 5 dB and $\gamma_{th}$ = 0 dB under different Rician fading factor $\kappa$.}\label{uplow2}
\end{figure}

\subsection{Upper and Lower Bounds of OP}
 In Fig.~\ref{uplow2}, we examine the simulation and asymptotic results of OP. In particular, the results assess the accuracy of the bounds. To present a comprehensive range of results, we consider both low SNR threshold $\gamma_{th}=5$ dB and high SNR threshold $\gamma_{th}=0$ dB. Upon observing the graph, the decreasing tendency of the OP is well imitated by the bounds. However, it can be noted that as the SNR threshold decreases, the gap between the upper bound and the exact value does not significantly increase, while the gap between the lower bound and the simulation results noticeably widens. This outcome is related to the selected lower bound approximation method and is consistent with the mathematical analysis presented in the Appendix C. Though a large gap is seen in the low SNR regime, the upper bound can be a conservative estimate of OP. Further examination of Fig.~\ref{uplow2} reveals that when the SNR threshold decreases, the lower bound and the simulation results are nearly identical. This indicates that in the high SNR regime, the lower bound can be used as an approximation for the OP, thereby serving as a useful metric for evaluating the system performance.

\begin{figure}[t!]
\centering
\includegraphics[width =3.5in]{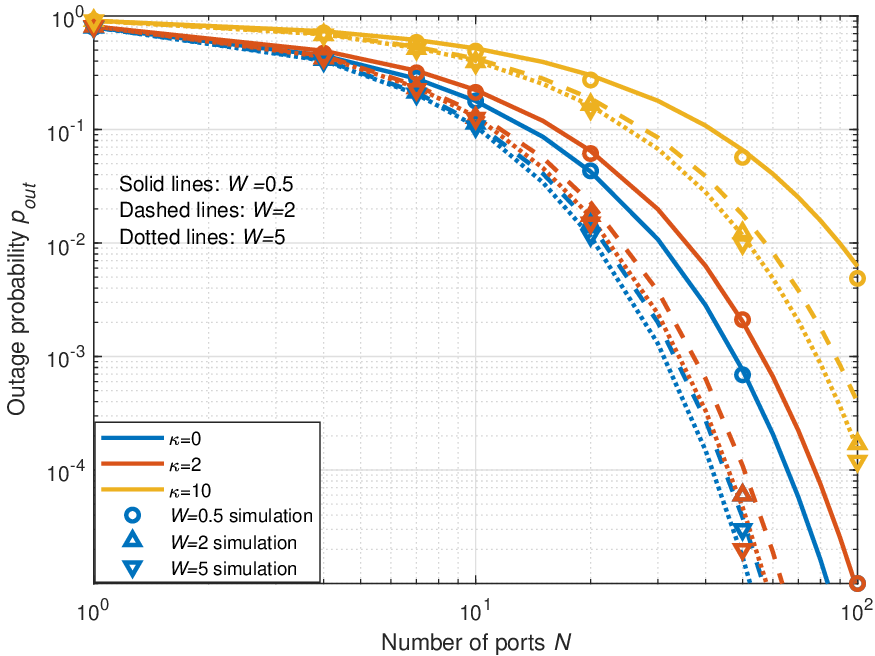}
\caption{The OP against the number of ports $N$ under different Rician fading factor $\kappa$ with $\gamma_{th}$ = 2 dB.}\label{opw}
\end{figure}

\subsection{Impact of FAS Size on OP} 
In Fig.~\ref{opw}, we investigate the impact of FAS size on the OP assuming ULA port configuration. Analytical and Monte Carlo simulation results are provided under various size of FAS $W\lambda$ in Fig.~\ref{opw}. In the results, the SNR threshold $\gamma_{th}$ is set to 2 dB. For different Rician fading factors \(\kappa\), we see that as the size of FAS increases, the OP decreases. This indicates that regardless of the ratio between the LoS and NLoS components, increasing the size of FAS can enhance the OP performance. Additionally, the results reveal that for different lengths of FAS, the OP consistently decreases as the number of ports \(N\) increases. It can be observed that for a 10-port FAS, when the Rician fading factor \(\kappa=0\), the OP can reach approximately $0.2$ with a space of $W=2\lambda$. If the space is increased to $W=5\lambda$ with 50 ports, the OP can be reduced to about \(1 \times 10^{-5}\). This demonstrates that for FAS, though space plays a role, achieving remarkable diversity within a small area is still feasible.

\begin{figure}[t!]
\centering
\includegraphics[width =3.5in]{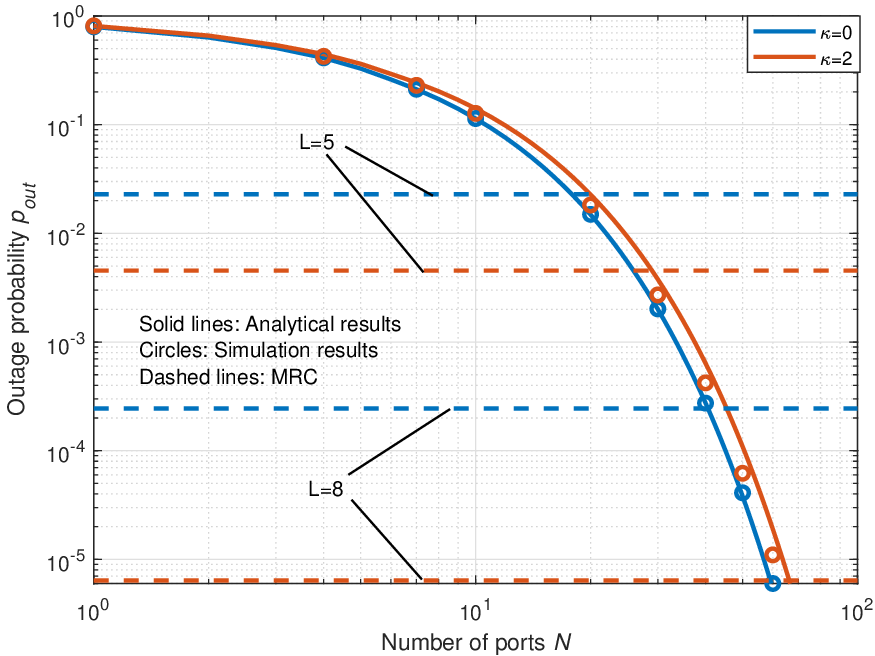}
\caption{The OP of FAS and MRC under different Rician fading factors $\kappa$ with $\gamma_{th}$ = 2 dB}\label{FASMRC}
\end{figure}

\subsection{Comparison of the Outage Performance between FAS and MRC} 
Fig.~\ref{FASMRC} shows both the analytical and simulation results for the OP of FAS and MRC system under various Rician fading factors \( \kappa \) and SNR thresholds \( \gamma_{th} \). The OP for the envelope of MRC with \( L \) branches under a Rician fading channel is expressed as \cite{MRC_Rician}:
\begin{equation}\label{MRC}
  p_{out}^{MRC} = 1 - {Q_L}\left( {\sqrt {2L\kappa } ,\sqrt {2\left( {\kappa  + 1} \right){\gamma _{th}}} } \right),
\end{equation}
where $Q_L$ is the $L$-order Marcum Q-function. In the ﬁgure, we provide the results for $L = 5, 8$.
It can be observed that as the number of ports $N$ increases, the OP of FAS decreases and becomes lower than that of the MRC. Given that the size of the FAS is set to $W\lambda= 2\lambda $, FAS can outperform MRC with 5 antennas if $N>30$. As the number of ports $N$ approaches 70, the FAS can outperform MRC with 8 antennas, which requires more physical space than FAS. Additionally, it is important to note that the number of RF chains in MRC must match the number of antennas, while in FAS, it always has one RF chain. This demonstrates that, FAS can achieve better performance than MRC by increasing the number of ports $N$ while requiring fewer RF chains.

\begin{figure}[t!]
\centering
\includegraphics[width =3.5in]{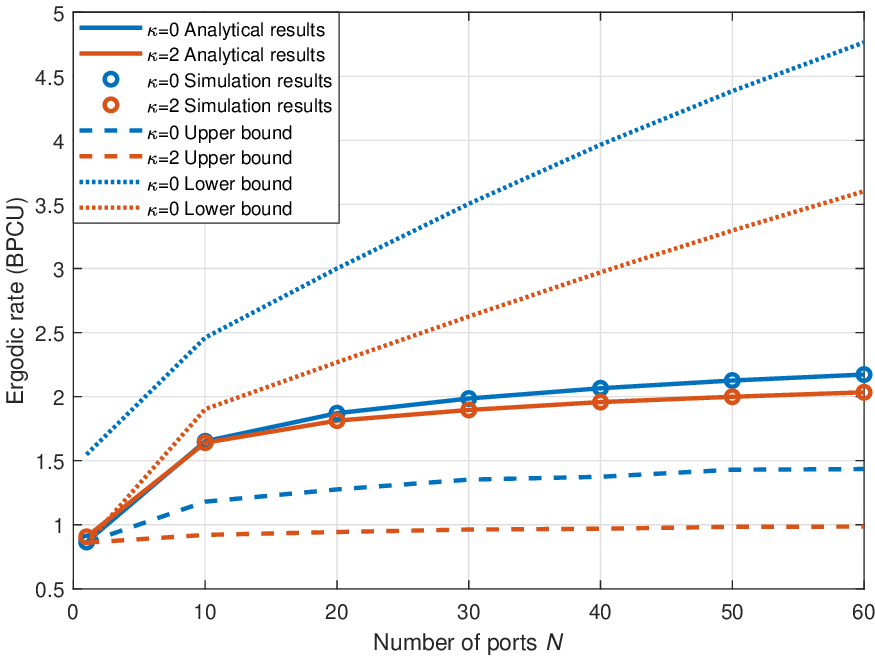}
\caption{The ER against the number of ports $N$ under different Rician fading factor $\kappa$.}\label{fig6}
\end{figure}

\subsection{Impact of the Number of Ports on ER}
Fig.~\ref{fig6} illustrates the ER of the Rx-SISO-FAS against the number of ports $N$. The solid lines, dotted lines, and dashed lines represent the simulation results, upper bounds, and lower bounds, respectively. We can observe that as the number of ports \(N\) increases, both the upper and lower bounds are good to imitate the rising trend of the ER when $N$ is small. However, for simulation results and the lower bound, when \(N\) continues to increase, the slope of the curves gradually decreases, while the upper bound develops a large gap. In contrast, the lower bound, although it also has a certain gap, shows a slope that is consistent with the simulation results, making it suitable for a conservative estimate of the system performance. It can also be seen that when the Rician fading factor $\kappa$ is small, the ER is relatively high. When the number of ports \(N\) reaches a certain value, the ER is maximum at $\kappa = 0$. 

\begin{figure}[t!]
\centering
\includegraphics[width =3.5in]{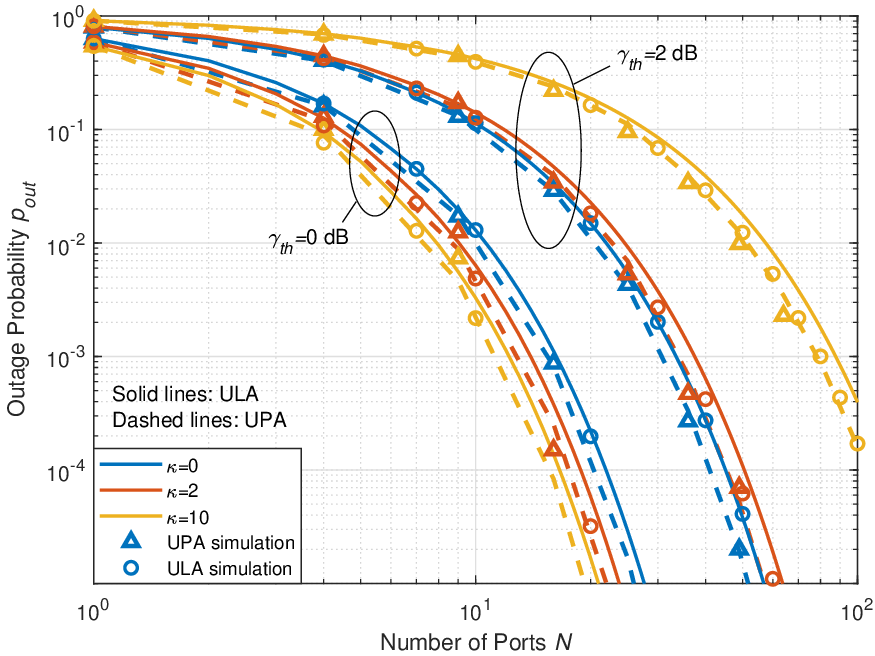}
\caption{The OP against the number of ports $N$ under ULA and UPA port configurations with $\gamma_{th}=0$ dB and $\gamma_{th}=2$ dB.}\label{upaop}
\end{figure}

\subsection{Impact of Port Configuration}
In Fig.~\ref{upaop}, we investigate the impact of the port configuration of FAS on the OP performance against the number of ports $N$. Fig.~\ref{upaop} presents the analytical and Monte Carlo simulation results for the OP of ULA and UPA configurations under different parameter settings. For comparison, both ULA and UPA configurations share the same number of ports. It is observed that when \(N\) is small, the change in port configuration does not significantly improve performance. This is because when the number of ports $N$ is sufficiently small, different port configurations have limited impact on the distances between the ports, resulting in only minor differences in spatial correlation. Consequently, the OP under two configurations are similar. However, as the number of ports $N$ increases, the difference in OP between the configurations becomes more pronounced. These results demonstrate that as the number of ports \(N\) increases, the impact of port configuration becomes more apparent, suggesting the possibility of utilizing more space for better OP performance. Finally, Fig.~\ref{fig8} studies the same but on the ER. It can be seen that as the number of ports \(N\) increases, the gap between the ER of UPA and that of ULA gradually widens.
 This difference becomes more obvious as the Rician fading factor decreases. 

\begin{figure}[t!]
\centering
\includegraphics[width =3.5in]{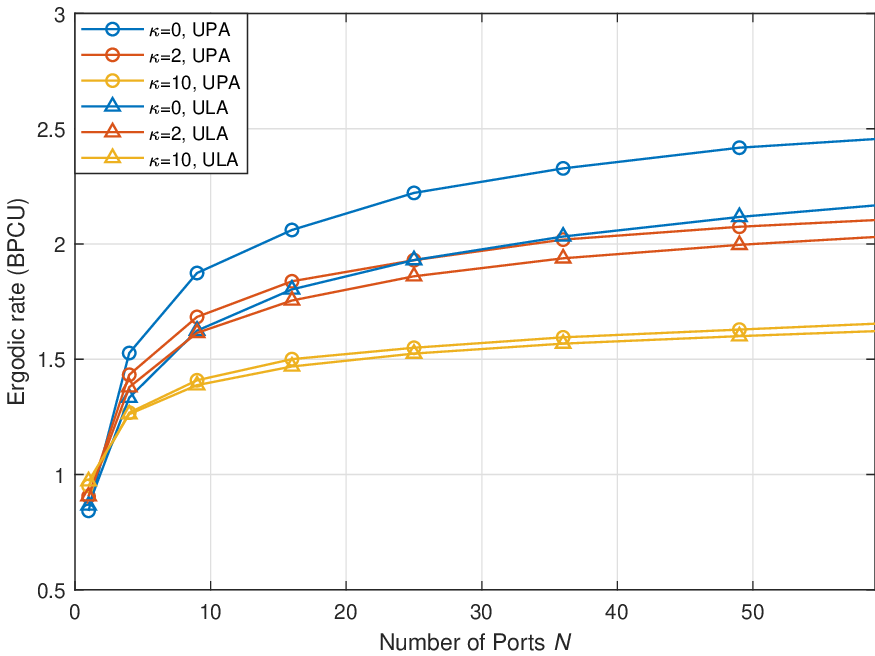}
\caption{The simulation results of ER against the number of ports $N$ under different Rician fading factor $\kappa$.}\label{fig8}
\end{figure}

\section{Conclusions}\label{sec:conclude}
In this paper, we analyzed the performance of Rx-SISO-FAS under Rician fading channels. We derived the exact OP expression and its closed-form upper and lower bounds, as well as the ER expressions and its closed-form upper bound. We then approximated the diversity order by the number of ports $N$ when the SNR was high enough. Additionally, we examined the impact of port configuration on the OP and ER system performance. The numerical results demonstrated that UPA port configuration outperforms ULA. Moreover, we benchmarked the outage performance of FAS against conventional MRC system for comparative analysis. Also, our numerical results indicated that the performance of FAS degraded significantly only in the low SNR regime with large Rician fading factor $\kappa$. Nevertheless, there are still many aspects of FAS that warrant further investigation to deepen our understanding. 

\numberwithin{equation}{section}
\section*{Appendix~A: Proof of Lemma~\ref{Lemma1:joint pdf}}\label{Appendix:As}
\renewcommand{\theequation}{A.\arabic{equation}}
\setcounter{equation}{0}
According to the channel model, conditioned on $x_0$, $y_0$, $|h_2|$ is Rician distributed~\cite{Rician_fading}, and we have 
\begin{equation}\label{A1}
\begin{aligned}
&{P_{\left| {{{\rm{h}}_2}} \right|\left| {{x_0},{y_0}} \right.}}\left( {{m_2}\left| {{x_0},{y_0}} \right.} \right) \\
&= \frac{{2{m_2}}}{{{\sigma ^2}\left( {1 - \rho _2^2} \right)}}{e^{ - \frac{{m_2^2 + \rho _2^2\left( {x_0^2 + y_0^2 + {A^2}} \right) + \left( {1 - \rho _2^2} \right){A^2}}}{{{\sigma ^2}(1 - \rho _2^2)}}}}\\
&\times{I_0}\left( {\frac{{2\sigma \sqrt {\rho _2^2\left( {x_0^2 + y_0^2 + {A^2}} \right) + \left( {1 - {\rho_2 ^2}} \right){A^2}} {m_2}}}{{{\sigma ^2}(1 - \rho _2^2)}}} \right).
\end{aligned}
\end{equation}
It can be observed that~\eqref{A1} only depends on $|h_1|$. For any port $n$ except port $1$, it can be rewritten as
\begin{equation}\label{A2}
\begin{aligned}
{P_{\left| {{{h}_n}} \right|\left| {\left| {{h_1}} \right|} \right.}}({m_n}\left| {{m_1}} \right.) &= \frac{{2{m_n}}}{{{\sigma ^2}(1 - \rho _n^2)}}{e^{ - \frac{{m_n^2 + \rho _n^2m_1^2 + \left( {1 - \rho _n^2} \right){A^2}}}{{{\sigma ^2}\left( {1 - \rho _n^2} \right)}}}}\\
&\times {I_0}\left( {\frac{{2\sqrt {\rho _n^2m_1^2 + \left( {1 - \rho _n^2} \right){A^2}} {m_n}}}{{{\sigma ^2}\left( {1 - \rho _n^2} \right)}}} \right).
\end{aligned}
\end{equation}
We can tell from \eqref{A2} that port $n$ is only related to port $1$. Given $h_1$, $h_2,\dots,h_N$ are all independent, we can obtain the conditional PDF given by
\begin{equation}\label{A3}
\begin{aligned}
&{P_{\left| {{h_2}} \right|,...,\left| {{h_N}} \right|\left| {\left| {{h_1}} \right|} \right.}}({m_1},{m_2},...,{m_N})\\
& = \prod\limits_{n = 2}^N {\frac{{2{m_n}}}{{{\sigma ^2}(1 - \rho _n^2)}}} {e^{ - \frac{{m_n^2 + \rho _n^2m_1^2 + (1 - \rho _n^2){A^2}}}{{{\sigma ^2}(1 - \rho _n^2)}}}}\\
&\times{I_0}\left( {\frac{{2\left( {\sqrt {\rho _n^2m_1^2 + (1 - \rho _n^2){A^2}} {m_n}} \right)}}{{{\sigma ^2}(1 - \rho _n^2)}}} \right).
\end{aligned}
\end{equation}
By multiplying \eqref{A3} with the PDF of $\left|h_1\right|$ given by (\ref{rician dis}), we have the joint PDF. For simplicity, we can express the PDF of $\left|h_1\right|$ in the form of \eqref{A2} by setting $\rho_1=0$ to obtain the desired result in \eqref{jpdf}, which completes the proof.

\numberwithin{equation}{section}
\section*{Appendix~B: Proof of Lemma~\ref{Lemma2:joint CDF}} \label{Appendix:Bs}
\renewcommand{\theequation}{B.\arabic{equation}}
\setcounter{equation}{0}
According to the joint PDF in (11), the joint CDF can be derived as
\begin{equation}\label{B1}
\begin{aligned}
&{F_{_{\left| {{h_1}} \right|,\left| {{h_2}} \right|,\dots,\left| {{h_N}} \right|}}}({m_1},{m_2},...,{m_N})\\
 &= P\left( {\left| {{h_1}} \right| < {m_1},\left| {{h_2}} \right| < {m_2},\dots,\left| {{h_N}} \right| < {m_N}} \right)\\
 &= \int_0^{{m_1}} {{\rm{\cdots}}} \int_0^{{m_N}} {{P_{\left| {{h_1}} \right|,\left| {{h_2}} \right|,\dots,\left| {{h_N}} \right|}}\left( {{t_1},{t_2},\dots,{t_N}} \right)} d{t_1}{\rm{\cdots}}d{t_N}\\
 &= \int_0^{{m_1}} {\frac{{2{t_1}}}{{{\sigma ^2}}}{e^{ - \frac{{t_1^2 + {A^2}}}{{{\sigma ^2}}}}}{I_0}\left( {\frac{{2{t_1}A}}{{{\sigma ^2}}}} \right)} \\
 &\times \prod\limits_{n = 2}^N {\int_0^{{m_n}} {\frac{{2{t_n}}}{{{\sigma ^2}(1 - \rho _n^2)}}{e^{ - \frac{{t_n^2 + \rho _n^2t_1^2 + (1 - \rho _n^2){A^2}}}{{{\sigma ^2}(1 - \rho _n^2)}}}}}}\\
 &\times{I_0}\left( {\frac{{2\left( {\sqrt {\rho _n^2t_1^2 + (1 - \rho _n^2){A^2}} {t_n}} \right)}}{{{\sigma ^2}(1 - \rho _n^2)}}} \right)  d{t_n}d{t_1}.
\end{aligned}
\end{equation}

The integral within the product operator represents the integration over the PDF of a Rician random variable. The LoS component \( v \) and the NLoS component \( \sigma_0 \) of the Rician random variable can be derived as follows:
\begin{equation}\label{B2}
\begin{aligned}
&\sigma _0^2 = {\sigma ^2}\left( {1 - \rho _n^2} \right)\\
&{v^2} = \rho _n^2m_1^2 + \left( {1 - \rho _n^2} \right){A^2}.
\end{aligned}
\end{equation}
We can represent the CDF of the Rician random variable by using the Marcum Q-function as
\begin{equation}\label{B3}
\begin{aligned}
F &= 1 - {Q_1}\left( {\frac{v}{{{\sigma _0}}},\frac{m}{{{\sigma _0}}}} \right)\\
&= 1 - {Q_1}\left( {\sqrt {\frac{{2\rho _n^2t_1^2}}{{{\sigma ^2}\left( {1 - \rho _n^2} \right)}} + 2\kappa} ,\sqrt {\frac{{2(1 + \kappa)}}{{{\sigma ^2}(1 - \rho _n^2)}}} m} \right).
\end{aligned}
\end{equation}

By substituting \eqref{B3} into \eqref{B1}, we have \eqref{B5} (see top of the next page).
\begin{figure*}
\begin{equation}\label{B5}
\begin{aligned}
&{F_{_{\left| {{h_1}} \right|,\left| {{h_2}} \right|,\dots,\left| {{h_N}} \right|}}}({m_1},{m_2},\dots,{m_N})\\
 &= \int_0^{{m_1}} {\frac{{2{t_1}}}{{{\sigma ^2}}}{e^{ - \frac{{t_1^2 + {A^2}}}{{{\sigma ^2}}}}}{I_0}\left( {\frac{{2{t_1}A}}{{{\sigma ^2}}}} \right)}  \prod\limits_{k = 2}^N {\left( {1 - {Q_1}\left[ {\sqrt {\frac{{2\rho _k^2t_1^2}}{{{\sigma ^2}\left( {1 - \rho _k^2} \right)}} + 2\kappa} ,\sqrt {\frac{{2(1 + \kappa)}}{{{\sigma ^2}(1 - \rho _k^2)}}} {m_k}} \right]} \right)} d{t_1}.
\end{aligned}
\end{equation}
\hrulefill
\end{figure*}
By changing the variable, the desired result is obtained.

\numberwithin{equation}{section}

\section*{Appendix~C: Proof of Corollary~\ref{Corollary1:lower bound}} \label{Appendix:Cs}
\renewcommand{\theequation}{C.\arabic{equation}}
\setcounter{equation}{0}
First, we determine the OP reduction for the $N$-th additional port to have~\eqref{C1} (see top of the next page). 
\begin{figure*}
\begin{equation}\label{C1}
\begin{aligned}
\Delta {p_{out}}\left( {{\gamma _{th}}} \right) &= {p_{out\left( {N - 1} \right)}} - {p_{out\left( N \right)}}\\
 &= \int_0^{{\gamma _{th}}} {{e^{ - \left( {\left( {\kappa  + 1} \right)t + \kappa } \right)}}{I_0}} \left( {2\sqrt {\kappa (\kappa  + 1)t} } \right){Q_1}\left( {\sqrt {\frac{{2\rho _N^2(\kappa  + 1)t}}{{(1 - \rho _N^2)}} + 2\kappa } ,\sqrt {\frac{{2(1 + \kappa )}}{{(1 - \rho _N^2)}}} \sqrt {{\gamma _{th}}} } \right)\\
 &\times \prod\limits_{n = 2}^{N - 1} {\left[ {1 - {Q_1}\left( {\sqrt {\frac{{2\rho _n^2(\kappa  + 1)t}}{{(1 - \rho _n^2)}} + 2\kappa } ,\sqrt {\frac{{2(1 + \kappa )}}{{(1 - \rho _n^2)}}} \sqrt {{\gamma _{th}}} } \right)} \right]} dt\\
 &= \left. {{Q_1}\left( {\sqrt {\frac{{2\rho _N^2(\kappa  + 1)t}}{{(1 - \rho _N^2)}} + 2\kappa } ,\sqrt {\frac{{2(1 + \kappa )}}{{(1 - \rho _N^2)}}} \sqrt {{\gamma _{th}}} } \right){p_{out\left( {N - 1} \right)}}} \right|_0^{{\gamma _{th}}}\\
 &- \int_0^{{\gamma _{th}}} {{{Q'}_1}\left( {\sqrt {\frac{{2\rho _N^2(\kappa  + 1)t}}{{(1 - \rho _N^2)}} + 2\kappa } ,\sqrt {\frac{{2(1 + \kappa )}}{{(1 - \rho _N^2)}}} \sqrt {{\gamma _{th}}} } \right)} {p_{out\left( {N - 1} \right)}}dt\\
& \le {Q_1}\left( {\sqrt {\frac{{2\rho _N^2(\kappa  + 1){\gamma _{th}}}}{{(1 - \rho _N^2)}} + 2\kappa } ,\sqrt {\frac{{2(1 + \kappa )}}{{(1 - \rho _N^2)}}} \sqrt {{\gamma _{th}}} } \right){p_{out\left( {N - 1} \right)}}
\end{aligned}
\end{equation}
\hrulefill
\end{figure*}
The second step in~\eqref{C1} can be proven using the monotonicity of the Marcum Q-function and the method of integration by parts in \cite{integral}. By rearranging the terms, we derive a recurrence formula, which allows us to establish the lower bound of OP as
\begin{equation}\label{C2}
\begin{aligned}
&{p_{out}}\left( {{\gamma _{th}}} \right) \ge {p_{out\left( 1 \right)}}\times\\
&\prod\limits_{n = 2}^N \left(1-{{Q_1}\left( {\sqrt {\frac{{2\rho _n^2(\kappa + 1)\gamma _{th}}}{{(1 - \rho _n^2)}} + 2\kappa} ,\sqrt {\frac{{2(1 + \kappa)}}{{(1 - \rho _n^2)}}} \sqrt {{\gamma _{th}}} } \right)}\right) .
\end{aligned}
\end{equation}
Finally, by substituting the expression in~\eqref{pout} when $N=1$ into~\eqref{C2}, we can complete the proof. It is important to note that as the upper limit of the integral \( \gamma_{th} \) increases, the omitted portion also increases, which may lead to a larger gap between the approximate and analytical results.

\section*{Appendix~D: Proof of Corollary~\ref{outup}} \label{Appendix:Ds}
\renewcommand{\theequation}{D.\arabic{equation}}
\setcounter{equation}{0}
According to \cite{marcum_appro}, in the case of $a<b$, the first-order Marcum Q-function can be approximated by
\begin{equation}\label{marcumq}
{Q_1}\left( {a,b} \right) \approx \sqrt {\frac{b}{a}} Q\left( {b - a} \right),
\end{equation}
in which $Q( \cdot )$ is the Gaussion Q-function with a lower bound given in \cite{Gauss_lower} as
\begin{equation}\label{D1}
\begin{aligned}
Q\left( x \right) &\ge \left( {\frac{{{e^{\frac{1}{{\left[ {\pi \left( {\kappa  - 1} \right) + 2} \right]}}}}}}{{2\kappa }}\sqrt {\frac{1}{\pi }\left( {\kappa  - 1} \right)\left[ {\pi \left( {\kappa  - 1} \right) + 2} \right]} } \right){e^{ - \frac{{\kappa {x^2}}}{2}}} \\
&= \alpha {e^{ - \frac{{\kappa {x^2}}}{2}}},
\end{aligned}
\end{equation}
where $\alpha  = \frac{{{e^{\frac{1}{{\left[ {\pi \left( {c  - 1} \right) + 2} \right]}}}}}}{{2c }}\sqrt {\frac{1}{\pi }\left( {c  - 1} \right)\left[ {\pi \left( {c  - 1} \right) + 2} \right]} $, and 
\begin{equation}
x = b - a= \sqrt {\frac{{2\left( {1 + \kappa} \right)}}{{\left( {1 - \rho _n^2} \right)}}} \left( {\sqrt {{\gamma _{th}}}  - \sqrt {\rho _n^2 + {A^2}\left( {1 - \rho _n^2} \right)} } \right).
\end{equation} 
By substituting the above into ~\eqref{D1} and letting \( t = \gamma_{th} \), we obtain
\begin{equation}\label{D3}
\begin{aligned}
&{Q_1}\left( {\sqrt {\frac{{2\rho _n^2(\kappa + 1)t}}{{(1 - \rho _n^2)}} + 2\kappa} ,\sqrt {\frac{{2(1 + \kappa)}}{{(1 - \rho _n^2)}}} \sqrt {{\gamma _{th}}} } \right)  \\
& > \frac{{\alpha {{\left[ {{\gamma _{th}}\left( {1 + \kappa } \right)} \right]}^{0.25}}}}{{\sqrt {\left| {{\rho _n}} \right|} {{\left[ {{\gamma _{th}}\left( {1 + \kappa } \right)} \right]}^{0.25}} + {{\left[ {\kappa \left( {1 - \rho _n^2} \right)} \right]}^{0.25}}}}\\
&\times {e^{ - \frac{c}{{1 - \rho _n^2}}\left( {{\gamma _{th}}\left( {\kappa  + 1} \right) + \kappa  - 2\sqrt {\kappa {\gamma _{th}}(\kappa  + 1)} } \right)}}\\
& = {\alpha _n}{e^{ - \frac{\kappa }{{1 - \rho _n^2}}{\gamma _n}}},
\end{aligned}
\end{equation}
By substituting \eqref{D3} into \eqref{pout}, we have the desired result.

\section*{Appendix~E: Proof of Corollary~\ref{erlow}} 
\label{Appendix:Es}
\renewcommand{\theequation}{E.\arabic{equation}}
\setcounter{equation}{0}
Due to the presence of a double integral, obtaining a closed form for the exact expression of ER is challenging. Therefore, it is necessary to approximate the original expression before integration. First, we derive the upper bound of the SNR's CDF using the method outlined in Appendix D as
\begin{equation}
\begin{aligned}\label{D}
F\left( x \right) &\le \left( {1 - {e^{ - x}}} \right)\prod\limits_{n = 2}^N {\left( {1 - {\alpha _n}{e^{ - {b_n}x_n}}} \right)} \\
& = 1 + \sum\limits_{n = 2}^N {\sum_{s \subseteq \left\{ {1,2,\dots,N} \right\}\atop \left| s \right| = n} {{{\left( { - 1} \right)}^n}\prod\limits_{i \subseteq s} {{a_i}} } } {e^{ - bi{x_n}}},
\end{aligned}
\end{equation}
where $b_n = \frac{c}{{1 - \rho _n^2}}$, and ${x _n} = {x}\left( {\kappa + 1} \right) + \kappa - 2\sqrt {\kappa{x}(\kappa + 1)} $.

We then simplify \eqref{D} by neglecting ${\kappa\left( {1 - \rho _n^2} \right)}$ in the denominator and $2\sqrt {\kappa{x}(\kappa + 1)}$ in the exponential so that
\begin{equation}
\hat F\left( x \right)  =  1 + \sum\limits_{n = 2}^N {\sum_{s \subseteq \left\{ {1,2,\dots,N} \right\}\atop \left| s \right| = n} {{{\left( { - 1} \right)}^n}\prod\limits_{i \subseteq s} {{a_i}} } } {e^{ - bix}}.
\end{equation}
By substituting the simplified $\hat F(x)$ into the expression of ER, we have
\begin{equation}\label{F3}
\begin{aligned}
&{\hat R_{N,\kappa }}\\
& = \frac{1}{{\ln \left( 2 \right)}}\int\limits_0^\infty  {\frac{{1 -\hat F\left( x \right)}}{{1 + x}}dx} \\
& = \frac{1}{{\ln \left( 2 \right)}}\sum\limits_{n = 2}^N {\sum_{s \subseteq \left\{ {1,2,\dots,N} \right\}\atop \left| s \right| = n} {{{\left( { - 1} \right)}^{n + 1}}\prod\limits_{i \subseteq s} {{a_i}} } } \int\limits_0^\infty  {\frac{{{e^{ - \sum\limits_{i \subseteq s} {{b_i}x} }}}}{{1 + x}}dx}.
\end{aligned}
\end{equation}
Note that in \cite{integral}, it is known that
\begin{equation}\label{F4}
\int_0^\infty  {\frac{{{e^{ - \mu x}}dx}}{{x + \beta }}}  =  - {e^{ - \mu \beta }}Ei\left( { - \mu \beta } \right).
\end{equation}
By performing simple changes of variables and substituting~\eqref{F4} into~\eqref{F3}, we obtain the desired result.

\bibliographystyle{IEEEtran}


\end{document}